\documentclass[peerreview]{IEEEtran}

\usepackage[square,numbers,sort&compress]{natbib}
\usepackage{authblk}
\usepackage{amsthm}
\usepackage{amsmath}
\usepackage{mathrsfs}
\usepackage{amssymb} 
\usepackage{physics}

\usepackage{enumerate}
\usepackage{bbm} 
\usepackage{graphicx}
\usepackage{verbatim}

\newcommand\blfootnote[1]{%
  \begingroup
  \renewcommand\thefootnote{}\footnote{#1}%
  \addtocounter{footnote}{-1}%
  \endgroup
}

\theoremstyle{remark}	\newtheorem{theorem}{Theorem}
\theoremstyle{remark}	\newtheorem{lemma}[theorem]{Lemma}
\theoremstyle{remark}	\newtheorem{corollary}[theorem]{Corollary}
\theoremstyle{remark}	
\theoremstyle{remark} \newtheorem{definition}{Definition}
\theoremstyle{remark} \newtheorem{remark}{Remark}
\theoremstyle{remark} \newtheorem{example}{Example}

\usepackage{xcolor}
\newcommand{\ie}{\emph{i.e.} }
\newcommand{\eg}{\emph{e.g.} }
\newcommand{\etal}{\emph{et al.} }
\newcommand{\cf}{\emph{cf.} }

\usepackage{hyperref} %
\usepackage%
{hyperref} %
\hypersetup{
    colorlinks,
    linkcolor={blue!80!black},%
    citecolor={green!50!black},
    urlcolor={blue!80!black}
}
\usepackage{accents}
\newlength{\dhatheight}
\newcommand{\doublehat}[1]{%
    \settoheight{\dhatheight}{\ensuremath{\hat{#1}}}%
    \addtolength{\dhatheight}{-0.35ex}%
    \hat{\vphantom{\rule{2pt}{\dhatheight}}%
    \smash{\hat{#1}}}}

\newcommand{\bieee}{\begin{IEEEeqnarray}{rCl}}
\newcommand{\eieee}{\end{IEEEeqnarray}}
\newcommand{\prob}[1]{\Pr\left(#1\right)}
\newcommand{\given}{\!\!\,\mid\,\!\!}
\newcommand{\cprob}[2]{\Pr\left(#1\given #2\right)}

\renewcommand{\mathbbm}[1]{\text{\usefont{U}{bbm}{m}{n}#1}} %
\newcommand{\eps}{\varepsilon}

\renewcommand{\trace}{\mathrm{Tr}}

\newcommand{\identity}{\mathbbm{1}}
\newcommand{\kb}[1]{ \ketbra{#1} } %

\newcommand{\hm}{\hat{m}}

\newcommand{\hM}{\hat{M}}

\newcommand{\tM}{\widetilde{M}}

\newcommand{\Aset}{\mathcal{A}}

\newcommand{\Dset}{\mathcal{D}}
\newcommand{\Eset}{\mathcal{E}}
\newcommand{\Fset}{\mathcal{F}}

\newcommand{\Hset}{\mathcal{H}}

\newcommand{\Pset}{\mathcal{P}}
\newcommand{\Uset}{\mathcal{U}}

\newcommand{\Tset}{\mathcal{T}}

\newcommand{\Xset}{\mathcal{X}}
\newcommand{\Yset}{\mathcal{Y}}
\newcommand{\Zset}{\mathcal{Z}}

\newcommand{\channel}{\mathcal{N}}

\newcommand{\inR}{\mathcal{R}}
\newcommand{\inQ}{\mathscr{L}}

\newcommand{\opC}{\mathcal{C}}
\newcommand{\opQ}{\mathcal{Q}}

\newcommand{\tset}{\Aset^{\delta}}													%

\begin{document}

\title{Communication with Unreliable Entanglement Assistance}
 
\author{\IEEEauthorblockN{Uzi Pereg,~\IEEEmembership{Member, IEEE}, Christian Deppe,~\IEEEmembership{Member, IEEE}, and 
Holger Boche,~\IEEEmembership{Fellow, IEEE}}
}

\maketitle

\begin{abstract}

Entanglement resources %
can increase transmission rates substantially. %
Unfortunately, entanglement is a fragile resource  that is quickly degraded by decoherence effects. 
In order to generate entanglement for optical communication, the transmitter and the receiver first prepare entangled spin-photon pairs locally, and then the photon at the transmitter is sent to the receiver through an optical fiber or free space. Without feedback, the transmitter does not know whether the entangled photon has reached the receiver.
The present work introduces a new model of unreliable entanglement assistance, whereby the communication system operates whether entanglement assistance is present or not. While the sender is ignorant, the receiver knows whether the entanglement generation was successful. %
In the case of a failure, the receiver decodes less information.  In this manner, the effective transmission rate is adapted according to the assistance status. 
Regularized formulas are derived for the classical and quantum capacity regions with unreliable entanglement assistance, characterizing the tradeoff between the unassisted rate and the excess rate that can be obtained from entanglement assistance.
It is further established that time division between entanglement-assisted and unassisted coding strategies is optimal for the noiseless qubit channel, but can be strictly suboptimal for a noisy channel.

\end{abstract}


\blfootnote{%
Uzi Pereg is with the Faculty of Electrical and Computer Engineering  and the Hellen Diller Quantum Center, 
Technion -- Israel Institute of Technology, 3200003 Haifa, Israel (email: uzipereg@technion.ac.il).
Christian Deppe is with the Institute for Communications Engineering,
Technische Universit\"at M\"unchen, 80333 Munich, Germany
(email: christian.deppe@tum.de).
Holger Boche is with the Institute of Theoretical Information Technology,
Technische Universit\"at M\"unchen, 80290 Munich, Germany, the Munich Center
for Quantum Science and Technology (MCQST), Schellingstr. 4, 80799
Munich, Germany, and the CASA – Cyber Security in the Age of Large-Scale
Adversaries – Excellenzcluster, Ruhr Universit\"at Bochum, Bochum, Germany
(email: boche@tum.de).}

\blfootnote{
Communicated by S. Mancini, Associate Editor for Quantum Information
Theory.
}

\blfootnote{
This paper was presented in part at  the 2022 IEEE International
Symposium on Information Theory, Espoo, Finland, June 26--July 1, 2022.
}

\section{Introduction}

Quantum channels represent the physical evolution of %
a non-isolated system %
and provide a mathematical description for  a noisy transmission medium, such as an optical fiber. %
The channel capacity is the ultimate characteristic for communication throughput, \ie the optimal transmission rate with an asymptotically vanishing error for a given noisy channel.
Generally speaking, quantum communication and security protocols can be categorized as either entanglement-assisted or unassisted.
Entanglement resources are instrumental %
in a wide variety of quantum network frameworks, such as physical-layer security \cite{YLLYCZRCLL:20p}, %
  interferometry \cite{KhabiboullineBorregaardDeGreveLukin:19p}, %
sensor networks \cite{ZhuangZhang:19p,XiaLiZhuangZhang:21p},   %
and communication complexity \cite{Adhikary:21p}. 
Furthermore, the data rate can be significantly higher when the communicating parties are provided with entangled particles \cite{BennettShorSmolin:99p,BennettShorSmolin:02p}, as has recently been demonstrated in experiments \cite{HaoShiLiShapiroZhuangZhang:21p}.
Unfortunately, entanglement is a fragile resource that is quickly degraded by decoherence effects \cite{JaegerSergienko:14p}. %

In order to generate entanglement in an optical communication system, the transmitter may %
prepare an entangled pair of photons locally, and then send one of them to the receiver \cite{YCLRLZCLLD:17p}. Such generation protocols are not always successful, as  photons are easily absorbed before reaching the destination. Therefore, practical systems require a back channel, to inform the transmitter whether the entanglement has been established to a satisfying degree of quality. In the case of a failure, the protocol is to be repeated. The backward transmission may result in a delay, which in turn leads to a further degradation of the entanglement resources. In this work, we propose a new principle of operation: Communication with unreliable entanglement assistance. In our model, the communication system operates on a  rate that is adapted to the status of the entanglement assistance, whether the assistance exists or not. Hence, feedback and repetition are not required.

Driven by new applications such as Industry 4.0, Vehicle-to-Everything  (V2X), and the Tactile Internet \cite{FBW14}, future communication systems such as those 
beyond the fifth generation of mobile networks (5G) will significantly differ from both existing wireless and wired networks. Quantum communication networks are expected to play an important role in the communication 
infrastructure of the modern digital society \cite{BBDFFJS:21b}. %
Such systems will have a more involved network structure and will impose more diverse and challenging quality-of-service (QoS) requirements on the network resilience and reliability, service availability, delay, security, privacy, and many others. 
Some of these
new requirements %
can only be met by using quantum communication \cite{FB21}. %
The deployment requirements will go beyond those of the traditional systems, \eg %
the Tactile Internet will allow not only the control of data, but also of physical and virtual objects. With such critical applications comes the need to address the trustworthiness of the system and its services.

Resilience and reliability are core elements of trustworthiness and have been identified as  key challenges for future communication systems 
\cite{FB22}.
Furthermore, resilience and reliability cannot necessarily be verified automatically on digital hardware, i.e. on Turing machines  \cite{BSP21b}. It is not Turing decidable whether an attacker can perform a denial-of-service attack or not. Thus, it is also not Turing decidable whether a communication system is trustworthy or not \cite{FB22}. Therefore, it is fundamentally important to achieve entirely new approaches for resilience by design and for reliability by design. Here, %
 we develop the theory for reliability by design for entanglement-assisted point-to-point quantum communication systems.

Communication through quantum channels can carry %
 classical or quantum information.
 For classical communication, the Holevo-Schumacher-Westmoreland (HSW) Theorem provides a regularized %
formula for the classical capacity of a quantum channel without assistance \cite{Holevo:98p,SchumacherWestmoreland:97p}. %
Although calculation of such a formula is intractable in general, it provides computable lower bounds, and there are special cases where the capacity can be computed exactly. 
The reason for this difficulty is that the Holevo information is super-additive \cite{Holevo:12b}. 
A similar difficulty occurs with transmission of quantum information. %
The regularized formula for the quantum capacity is given 
in terms of the coherent information \cite{%
Devetak:05p}. 
 The entanglement-assisted classical capacity and quantum capacity of %
a noisy quantum channel were fully characterized by Bennett \etal 
\cite{BennettShorSmolin:99p,BennettShorSmolin:02p} in terms of the quantum mutual information, in analogy  
to Shannon's  capacity formula for a classical channel \cite{Shannon:48p}. 
The tradeoff between transmission, leakage, key,   and entanglement rates is studied extensively in the literature as well \cite{Shor:04p,DevetakHarrowWinter:08p,HsiehWilde:10p,HsiehWilde:10p1,WildeHsieh:12p1,WangHayashi:20c,ZhuangZhuShor:17p,NoetzelDiAdamo:20c,PeregDeppeBoche:21p}.

The theory of uncertain cooperation was first introduced to classical information theory 
in 2014
by Steinberg \cite{Steinberg:14c}, and further investigated by Huleihel and Steinberg \cite{HuleihelSteinberg:17p}. The motivation is based on the 
engineering aspects of modern communication networks.
In a dynamic ad-hoc communication setup, the availability of resources, such as bandwidth, time slots, and energy, is not guaranteed a priori, since their availability depends on parameters that the network designer does not control. For example, cooperation can depend on the battery status of intermediate users (relays or repeaters), on weather, or simply the willingness of peers in the network to help. A typical situation, therefore, is that the users are aware of the possibility that cooperation will take place, but it cannot be assured before transmission begins.  
Our framework is inspired by Steinberg's model \cite{Steinberg:14c}.
The classical models of unreliable cooperation mainly focus on dynamic resources in multi-user settings, such as the 
multiple-access channel \cite{HuleihelSteinberg:16c} and the broadcast channel \cite{ItzhakSteinberg:17c,ItzhakSteinberg:21p,PeregSteinberg:21c4}.
Other approaches for unreliable communication links include the outage analysis \cite{OzarowShamaiWyner:94p,KarasikSimeoneShamai:13p}, automatic repeat request (ARQ) \cite{CaireTuninetti:01p,SteinerShamai:08p}, and cognitive radios \cite{GJMS:09p}.
Our focus here, however,  is on a point-to-point quantum channel and the reliability of 
static resources.

In this paper, we consider communication of either classical or quantum information over a quantum channel, while Alice and Bob are provided with \emph{unreliable} entanglement resources, as the communicating parties are uncertain about the availability of entanglement assistance.
Specifically,  Alice wishes to send two messages, at rates $R$ and $R'$. She encodes both messages using her share of the entanglement resources, as she does not know whether Bob will have access to the entangled resources.
Bob has two decoding procedures. If the entanglement assistance has failed to reach Bob's location, he performs a decoding operation to recover the first message alone. Hence, the communication system operates on a rate $R$.
Whereas if Bob has entanglement assistance, he decodes both messages, hence the overall transmission rate is $R+R'$.
In other words, $R$ is a \emph{guaranteed rate}, and $R'$ is the \emph{excess rate} of information that entanglement assistance provides.
We define the capacity region as the set of all rate pairs $(R,R')$ that can be achieved 
with asymptotically vanishing decoding errors. We establish a regularized characterization for the classical and quantum capacity regions. 
We are developing reliability by design. 
If the entanglement resource is unreliable, then  the rate $R$ can be guaranteed regardless. %
For the applications mentioned above, this is of central importance, because %
absolutely critical data can be transmitted at rate $R$ 
and the communication does not break down.

The communication design makes a compromise. 
In general terms, the extreme options are to use the entanglement resources to the fullest extent, or ignore them completely.
Those extreme strategies attain the corner points of the capacity region.
The former strategy achieves the point
$(0,R')$, and the latter achieves 
$(R,0)$.
A communication protocol that relies heavily on the entanglement resources reaps the benefits of entanglement to a high extent, if the assistance is present. However, if the entanglement generation fails, then the transmission rate will be very low.
That is, the excess rate $R'$ will be close to optimal, while the guaranteed rate $R$ will be  low.
 If the designer decides to sacrifice excess rate and reduce $R'$, \ie reduce %
the gain from the entanglement resources, then %
we can guarantee a higher transmission rate. Our results characterize the optimal tradeoff.

\begin{figure}[tb]
\center
\includegraphics[scale=0.5,trim={3.5cm 9.5cm 0 9cm},clip]{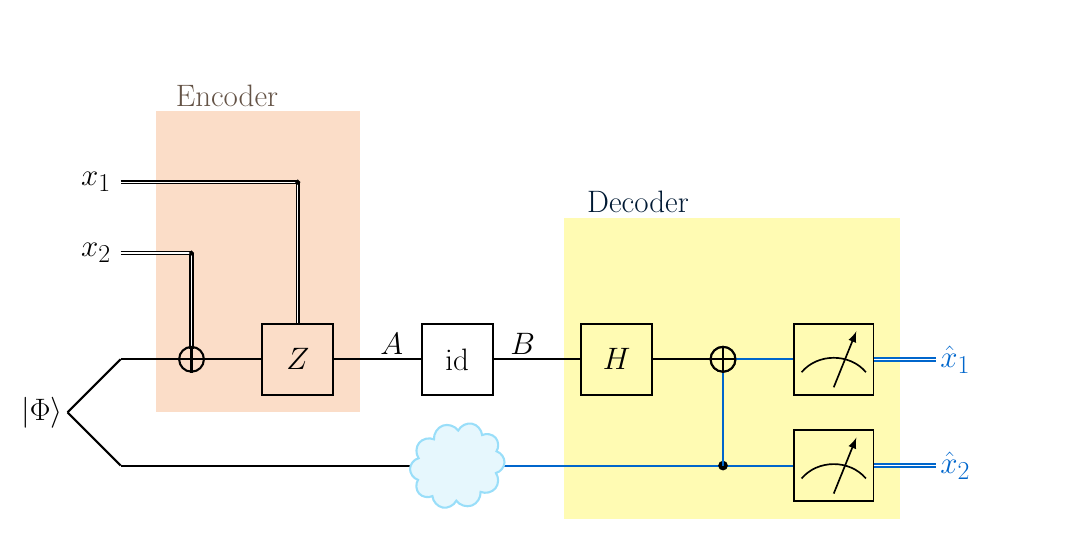} 
\caption{Super-dense coding with unreliable entanglement assistance. The blue lines indicate the bits and qubits that are affected when the  entanglement resources fail to reach Bob's location.}
\label{fig:DenseUncertain}
\end{figure}

Consider the simple scenario of a noiseless qubit channel $\text{id}_{A\to B}$, for which we have two elementary communication methods:
\begin{enumerate}
\item
Send one classical bit of information.
\item
Employ the super-dense coding protocol in order to send two classical bits, as illustrated in Figure~\ref{fig:DenseUncertain}.
\end{enumerate}
The first method is optimal without assistance, while the second is optimal when entanglement assistance is present.
If Alice follows the super-dense coding protocol, but the entanglement resources do not reach Bob's location, then Bob measures a qubit that has no correlation with the information bits. In the framework of unreliable entanglement assistance, 
Method 2 achieves a zero guaranteed rate and an excess rate of two information bits per transmission. %
Suppose that Alice employs time division: She sends $(1-\lambda) n$ transmissions using Method 1, and $\lambda n$ transmissions following Method 2, where $0\leq\lambda\leq 1$.
Hence, the communication system operates on a guaranteed rate of $R=1-\lambda$ information bits per transmission, and an excess rate of $R'=2\lambda$ information bits per transmission. We show that the time division region is optimal for the noiseless qubit channel. Nevertheless, we demonstrate that time division can be strictly suboptimal for a very simple noisy channel.

The analysis is based on a novel method that is inspired by the classical network technique of superposition coding (SPC) \cite{CoverThomas:06b}. 
Originally,
the classical SPC scheme consists of a collection of sequences $u^n(m)$ and 
$v^n(m,m')$ of length $n$, where $m$ and $m'$ are messages that are associated with different users in a multi-user network. 
The sequences $u^n(m)$ are called cloud centers, while $v^n(m,m')$ are displacement vectors, and the codewords $c^n(m,m')=u^n(m)+v^n(m,m')$ are thought of as satellites.
In analogy, we use conditional quantum operations that map quantum cloud centers to quantum satellite states.
Suppose that Alice and Bob share an entangled state $\varphi$ a priori.
Each cloud center is associated with a classical sequence $x^n(m)$, and
at the center of each cloud there is a state, $ \sigma_m=\Fset^{(x^n(m))}(\varphi)$, where $\Fset^{(x^n(m))}$ is a quantum encoding map that is conditioned on $x^n(m)$.  Applying random Pauli operators that 
encode the message $m'$  takes us from the cloud center to a satellite $\rho_{m,m'}$ on the cloud that depends on both messages, $m$ and $m'$. The channel input is the satellite state. 
See Figure~\ref{fig:Superposition}.
Bob decodes in two steps. First, Bob recovers the cloud, \ie he estimates $m$.
If the entanglement assistance is absent, then Bob quits after the first step.
Otherwise, if Bob has entanglement assistance, then he continues to decode the satellite $m'$.
We show 
that even in our fundamental point-to-point setting, quantum SPC can outperform time division.

\begin{figure}[tb]
\center
\includegraphics[scale=0.85,trim={0.5cm 0.5cm 0 -0.1cm},clip]{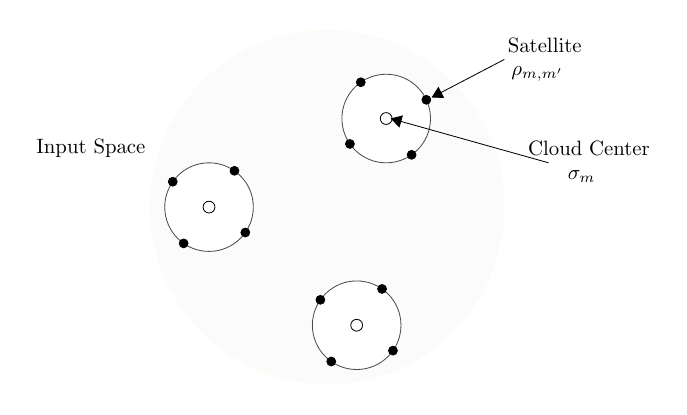} 
\caption{Superposition coding (SPC): The quantum version.}
\label{fig:Superposition}
\end{figure}

\section{Definitions and Related Work}
\subsection{Notation %
and Information Measures}
The quantum state of a system $A$ is a density operator $\rho$ on the Hilbert space $\Hset_A$.
The set of all such density operators is denoted by $\mathscr{S}(\Hset_A)$. 
A measurement of a quantum system is a set of operators $\{\Lambda_j \}$ that forms a positive operator-valued measure (POVM), \ie
$\Lambda_j\succeq 0$ and 
$\sum_j \Lambda_j=\identity$, where $\identity$ is the identity operator. According to the Born rule, if the system is in state $\rho$, then the probability of the measurement outcome $j$ is given by $%
\trace(\Lambda_j \rho)$.
The trace distance between two density operators $\rho$ and $\sigma$ is $\norm{\rho-\sigma}_1$, where $\norm{F}_1=\trace(\sqrt{F^\dagger F})$.

Given a bipartite state $\rho_{AB}$ on $\Hset_A\otimes \Hset_B$, 
define the quantum mutual information as
\begin{align}
I(A;B)_\rho=H(\rho_A)+H(\rho_B)-H(\rho_{AB}) \,,
\end{align} 
where $H(\rho) \equiv -\trace[ \rho\log(\rho) ]$ is the von Neumann entropy of the state $\rho$.
Furthermore, conditional quantum entropy and mutual information are defined by
$H(A| B)_{\rho}=H(\rho_{AB})-H(\rho_B)$ and
$I(A;B|  C)_{\rho}=H(A|  C)_\rho+H(B|  C)_\rho-H(A,B|  C)_\rho$, respectively.
The coherent information is then defined as
\begin{align}
I(A\rangle B)_\rho=-H(A|  B)_\rho  \,.
\end{align}
The maximally entangled state %
between two systems %
of dimension $d$ %
is denoted by
$%
\ket{ \Phi_{AB} } = \frac{1}{\sqrt{d}} \sum_{j=0}^{d-1} \ket{ j}_A\otimes \ket{j}_B %
$, where $\{ \ket{j}_A \}$ and $\{ \ket{j}_B \}$  %
are respective orthonormal bases.

 We also use the following notation conventions. %
Calligraphic letters $\Xset$, $\Yset$, $\Zset$, $\ldots$ are used for finite sets.
Lowercase letters $x,y,z,\ldots$  represent constants and values of classical random variables, and uppercase letters $X,Y,Z,\ldots$ represent 
random variables.  
 We use $x^j=(x_1,x_{2},\ldots,x_j)$ to denote  a sequence of letters from $\Xset$, and $[i:j]$ for the index set
$\{i,i+1,\ldots,j\}$, where $j>i$. 

\subsection{Quantum Channel}
\label{subsec:Qchannel}
A quantum channel maps a state at the sender system to a  state at the receiver system. 
Formally, a quantum channel $\channel_{A\to B}:\mathscr{S}(\Hset_A)\to \mathscr{S}(\Hset_B)$ is defined by a   linear, completely positive, trace preserving map 
$%
\channel_{A\to B}  %
$. In the Stinespring representation, a quantum channel is specified by $%
{\channel}_{A\to B}(\rho_{A})=\trace_E( U\rho_{A} U^\dagger )
$, %
where the operator $U$ is an isometry, \ie $ U^\dagger U=\identity$. 
We assume that the quantum channel has a product form:  If  $A^n=(A_1,\ldots,A_n)$ are sent through $n$ channel uses, then the input state $\rho_{A^n}$ undergoes the tensor product mapping $\channel_{A^n\to B^n}\equiv  \channel_{A\to B}^{\otimes n} $. 
The sender and the receiver are often referred to as Alice and Bob.

\begin{figure}[bt]
 \centering
\hspace{-1.5cm}
\begin{minipage}{.4\textwidth}
\includegraphics[scale=0.5,trim={1cm 10.5cm 4cm 11cm},clip]{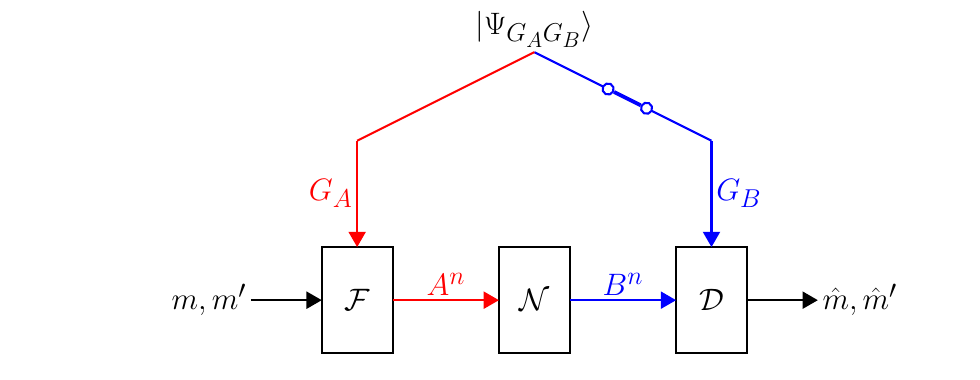} 

 \centering
\hspace{1.75cm}(a)
\end{minipage}
\hspace{1.5cm}
\begin{minipage}{0.4\textwidth}
\includegraphics[scale=0.5,trim={1.25cm 10.5cm 4cm 11cm},clip]{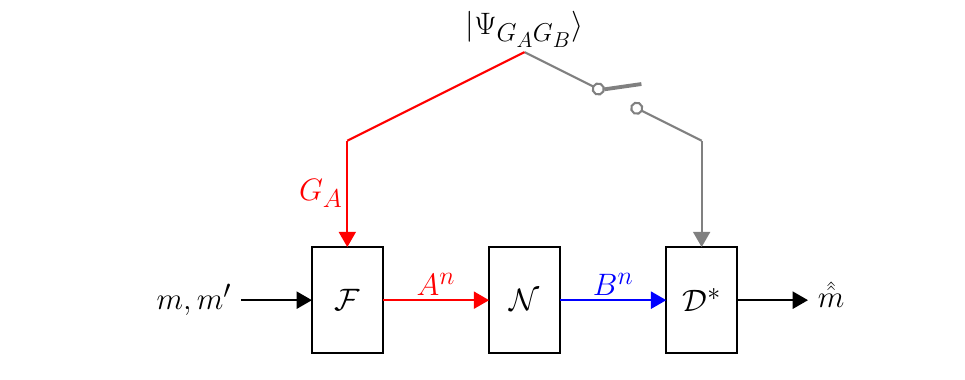} 

 \centering
\hspace{1.75cm}(b) 
\end{minipage}
\caption{Illustration of unreliable entanglement assistance that is controlled by an imaginary switch. The quantum systems of Alice and Bob are marked in red and blue, respectively.
 Alice encodes the messages $m$ and $m'$ by applying the encoding map $\Fset^{m,m'}_{  G_A \to A^n}$ to the system $G_A$, which is entangled with $G_B$. The model assumes that entanglement assistance may fail to reach Bob. Thus, there are two scenarios: 
(a) ``On": Bob performs a measurement $\Dset_{B^n G_B}$ in order to estimate $m$ and $m'$. (b) ``Off": Bob performs a measurement 
$\Dset^*_{B^n}$ to estimate $m$, and does not recover $m'$, as he cannot access $G_B$.  }
\label{fig:EAsiCode}
\end{figure}

\subsection{Coding with Unreliable Assistance}
\label{subsec:Mcoding}
We give coding definitions for communication with 
unreliable entanglement resources. We denote Alice and Bob's entangled systems by
$G_A$ and $G_B$, respectively.

\subsubsection{Classical Codes}
\begin{definition} %
\label{def:EAcapacity}
A $(2^{nR},2^{nR'},n)$  classical code with unreliable entanglement assistance consists of the following:   
Two message sets $[1:2^{nR}]$ and $[1:2^{nR'}]$, where $2^{nR}$, $2^{nR'}$ are assumed to be integers, a pure entangled state $\Psi_{G_{A},G_{B}}$, 
  a collection of encoding maps $\Fset^{m,m'}_{G_{A}\rightarrow A^n}:\mathscr{S}(\Hset_{G_A})\to \mathscr{S}(\Hset_{A}^{\otimes n})$ for $m\in [1:2^{nR}]$ and $m'\in [1:2^{nR'}]$,  and two decoding POVMs 
	$\Dset_{B^n G_B}
	=\{ D_{m,m'} \}$ and $\Dset^*%
	_{B^n}=\{ D^*_{m} \}$.
We denote the code by $(\Fset,\Psi,\Dset,\Dset^*)$.

The communication scheme is depicted in Figure~\ref{fig:EAsiCode}. The sender Alice has the systems $G_A,A^n$ and the receiver Bob has the system $B^n$, and \emph{possibly} $G_B$ as well, where $G_A$ and $G_B$ are entangled.   The model captures two scenarios, \ie when entanglement assistance is present or absent. This is illustrated in Figure~\ref{fig:EAsiCode} by an imaginary switch that controls the assistance. Without assistance, Bob is only required to decode one message, and given entanglement assistance, he should recover both messages. 

Specifically,  Alice chooses two classical messages, $m\in [1:2^{nR}]$ and $m'\in [1:2^{nR'}]$. %
She applies the encoding channel $\Fset^{m,m'}_{G_{A}\to A^n}$ to her share of the entangled state 
$\Psi_{G_{A},G_{B}}$, and then transmits %
$A^n$ over $n$ channel uses of $\channel_{A\rightarrow B}$. 
 Bob receives the channel output %
$B^n$. 
If the entanglement assistance is present,  \ie Bob has access to the entanglement resource $G_B$, then he should recover both messages. He combines the output with the entangled system $G_B$, and performs the POVM
 $\Dset_{B^n G_B}=\{ D_{m,m'}  \}$ to obtain an estimate $(\hm,\hm')$.

Otherwise, if entanglement assistance is absent, then Bob does not have $G_B$, and he is only required to recover $m$.   Hence,
 he performs the measurement $\Dset_{B^n}=\{ D_m^* \}$ to obtain an estimate $\doublehat{m}$ of the first message alone.
In the presence of entanglement assistance, the conditional probability of error given that the messages $m$ and $m'$ were sent, is  
\begin{align}
&P_{e|  m,m'}^{(n)}(\Fset,\Psi,\Dset)= 1- %
\trace\big[ D_{m,m'}
(\channel^{\otimes n}_{A\rightarrow B}\otimes\text{id})
(\Fset^{m,m'}\otimes\text{id}) (\Psi_{G_A,G_B})
 \big] %
\intertext{and without assistance,}
&P_{e|  m,m'}^{*(n)}(\Fset,\Psi,\Dset^*)= 
1-
\trace\Big[ D_{m}^*\,
\channel^{\otimes n}_{A\rightarrow B}\,
\Fset^{m,m'}(\Psi_{G_A})  \Big] 
\,.
\end{align}
Notice that the encoded input remains the same,  since Alice does not know whether entanglement assistance is present or not. Therefore, the  error depends on $m$ and $m'$ in both cases.

Given $\eps>0$, we say that the code %
is
a $(2^{nR},2^{nR'},n,\eps)$ classical 
code %
if the error probabilities are bounded by $\eps$. That is, %
$P_{e|  m,m'}^{(n)}(\Fset,\Psi,\Dset)\leq\eps $ and $P_{e|  m,m'}^{*(n)}(\Fset,\Psi,\Dset^*)\leq\eps $
for all $m\in [1:2^{nR}]$ and $m'\in [1:2^{nR'}]$. %
A rate pair $(R,R')$ is called achievable  if for every $\eps>0$ and sufficiently large $n$, there exists a $(2^{nR},2^{nR'},n,\eps)$ 
code with unreliable entanglement assistance. 

The classical capacity region $\opC_{\text{EA*}}(\channel)$ with unreliable entanglement assistance is defined as the set of  achievable rate pairs.
\end{definition}

\begin{remark}
Our model accounts for two extreme cases, \ie either the \emph{entire} entanglement resources are available or not at all.
In digital communications, this approach is referred to as a
hard decision \cite{Proakis:01b}.
Here, the decoder performs a hard decision on whether the entanglement resources are usable or not.
In analogy to the classical cooperation model 
\cite{HuleihelSteinberg:17p},
our model is based on the 
engineering aspects and the architecture of modern communication networks.
We expect that quantum communication networks in the future will follow
similar reliability guidelines. 
In particular, we envision that in a large quantum communication
network, the availability of entanglement resources will  not be guaranteed a priori. For example, entanglement resources will depend on physical conditions such as the weather, on the  status of quantum repeaters, or  the willingness of peers 
to help. In such a network, the transmitter and the receiver are aware of the \emph{possibility} that entanglement assistance will be available, but it cannot be assured before transmission begins. 
\end{remark}

\begin{remark}
In practical systems, heralded entanglement generation guarantees that Bob knows whether the procedure was successful or not. 
Thus, our assumption that Bob knows whether the entangled resource is present or absent is a practical one.
Specifically, in optical communication, both Alice and Bob prepare an entangled photon pair or spin-photon pair locally, see Figure~\ref{fig:HeraldedEnt}. Let us denote the pairs by $|\Phi_{G_A P_A}\rangle$ and $|\Phi_{G_B P_B}\rangle$, respectively, where  $P_A$ and $P_B$ represent photons. In order to generate entanglement with Bob, Alice transmits the photon $P_A$. 
If the photon transmission was successful, then Bob has the two photons $P_A$ and $P_B$ in his lab, as well as the quantum system $G_B$. 
In this case,  a Bell measurement on $P_A$ and $P_B$ eliminates the photons, but the remaining systems of Alice and Bob, $G_A$ and $G_B$, become entangled. If the photon has not reached Bob, then 
 the measurement outcome indicates so.
\end{remark}

\begin{remark}
It is important to note that unreliable assistance is \emph{not} equivalent to noisy assistance, which was considered by Zhuang \etal \cite{ZhuangZhuShor:17p}.
In particular, we do not associate a statistical model to the availability of the entanglement resources.
Instead, we consider a rate region that reflects the tradeoff between the guaranteed rate and the excess rate.
The guaranteed rate $R$ corresponds to information that Bob recovers whether the entanglement assistance is present or not, while the excess rate $R'$ represents the additional information that is sent if entanglement assistance is present. 
In other words, the rate $R$ represents the worst-case scenario, whereas $R'$ is associated with the best-case scenario.
As opposed to the average performance that is considered in the statistical model \cite{ZhuangZhuShor:17p}, we provide a worst-case best-case performance analysis. In the next remark, we give an illustration. 
\end{remark}

\begin{remark}
\label{remark:Ship_Illustration}
To illustrate the reliability approach,
consider the following metaphor.
$N$ travelers are
embarking on a long journey on a ship that may have a varying number of lifeboats. Overall, the lifeboats can accommodate $L$ travelers. 
In the event that the ship sinks, $(N-L)$ travelers will be rescued and brought back to their starting point, and the journey will continue with the remaining travelers in the lifeboats. The speed of the ship is $V=V(N,L)$, and the speed of the lifeboats is $v_0$. If the ship does not sink, the speed of each traveler will be $V$. However, if the ship sinks, the travel speed will be calculated as the average speed of the lifeboats, $R=(L/N)v_0$. In this case, $R$ represents the guaranteed speed of travel for the remaining travelers, and $R'=V-R$ represents the excess speed that the ship would have provided. Using more lifeboats will increase the guaranteed speed of travel, but decrease the excess speed, while using fewer lifeboats will have the opposite effect. It is important to consider both speeds, $R$ and $R'$, rather than just the average speed, when planning for the worst-case scenario of the ship sinking.
In the result section, we will discuss the option of dividing the travelers between different ships, and the advantage that may arise if a traveler can be in a quantum superposition between two ships.
\end{remark}

\begin{figure}[tb]
\centering
\includegraphics[scale=0.3,angle=90,trim={7cm 1cm 5cm 5cm},clip
]{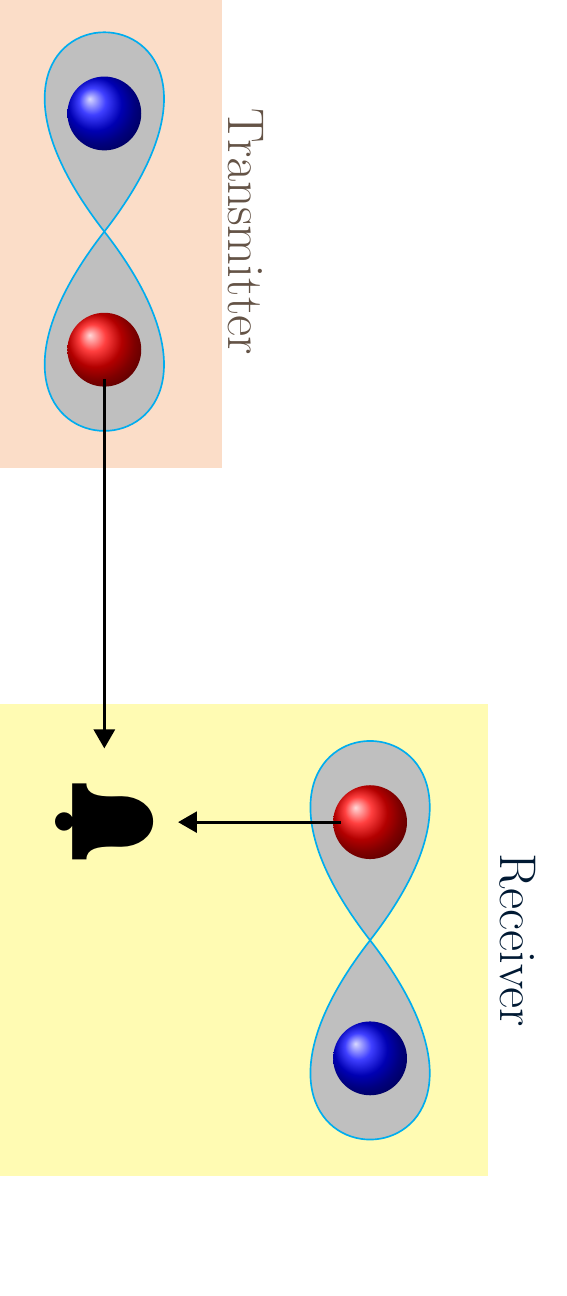} %
\caption{Heralded entanglement generation.}
\label{fig:HeraldedEnt}
\end{figure}

\subsubsection{Quantum Codes}
Next, we give a definition of a quantum code with unreliable entanglement assistance.
\begin{definition} 
\label{def:QEAcapacity}
A $(2^{nQ},2^{nQ'},n)$ quantum code with unreliable entanglement assistance consists of the following:  
A product Hilbert space $\Hset_M\otimes \Hset_{\bar{M}}$ %
with dimensions $\abs{\Hset_M}=2^{nQ}$ and $\abs{\Hset_{\bar{M}}}=2^{n(Q+Q')}$, %
 a pure entangled state $\Psi_{G_A,G_B}$,
 an encoding channel $\Fset_{G_{A} M \bar{M}\rightarrow A^n}:\mathscr{S}(\Hset_{G_A}\otimes \Hset_M\otimes
 \Hset_{\bar{M}})\to
\mathscr{S}(\Hset_A^{\otimes n})$, and two decoding channels 
$ \Dset_{B^n G_B\rightarrow \tM } :\mathscr{S}(\Hset_B^{\otimes n}\otimes \Hset_{G_B})\to\mathscr{S}(\Hset_{\bar{M}}) $ and $ \Dset^*_{B^n\rightarrow \hM}:\mathscr{S}(\Hset_{B}^{\otimes n})\to\mathscr{S}(\Hset_M)$.

The sender Alice has the systems $G_A,M,\bar{M},A^n$ and the receiver Bob has the systems $B^n,\hM,\tM$, and possibly $G_B$, where $G_A$ and $G_B$ are entangled. We think of  $M$ and $\bar{M}$ as quantum message systems.
Alice has a product state $\theta_M\otimes \xi_{\bar{M}}$. Let $\ket{\theta_{MK}}\otimes \ket{\xi_{\bar{M}\bar{K}}}$ be a purification of Alice's state, while $K$, $\bar{K}$ are arbitrary purifying systems.
Alice encodes the input state by applying the encoding channel 
$\Fset_{G_A M \bar{M} \to A^n}$ to $M$, $\bar{M}$, and to her share of the entangled state $\Psi_{G_{A},G_{B}}$,  and transmits the system $A^n$ over $n$ channel uses of 
$\channel_{A\to B}$. Bob receives the channel output systems $B^n$.
 If the entanglement assistance is present, then he combines the output
with the entangled system $G_B$, and applies the decoding channel
 $\Dset_{B^n G_B\to \tM}$. 
Otherwise, if entanglement assistance is absent, then he performs $\Dset^*_{B^n\to\hM}$.

Given $\eps>0$, the code %
is said to be a $(2^{nQ},2^{nQ'},n,\eps)$ quantum code with unreliable entanglement assistance if the trace distance between the original state and the resulting state at the receiver is bounded by $\eps$ in each scenario, \ie
\begin{align}
\frac{1}{2}\norm{\xi_{\bar{M}\bar{K}} -\Dset \channel^{\otimes n}_{A\to B} \Fset  \left( \theta_{MK}\otimes\xi_{\bar{M}\bar{K}}\otimes \Psi_{G_A,G_B} \right)  }_1
&\leq \eps \,,
\intertext{and}
\frac{1}{2}\norm{\theta_{MK} -\Dset^* \channel^{\otimes n}_{A\to B} \Fset  \left( \theta_{MK}\otimes\xi_{\bar{M}\bar{K}}\otimes \Psi_{G_A} \right)  }_1
&\leq \eps \,,
\end{align}
where $\norm{\cdot}_1$ denotes the trace norm. Observe that the second error depends on the entangled state only through the reduced state of $G_A$, since the receiver does not have access to $G_B$ in the scenario of absent assistance.
A rate pair $(Q,Q')$ is said to be  achievable if for every $\eps>0$ and sufficiently large $n$, there exists a $(2^{nQ},2^{nQ'},n,\eps)$
code with unreliable entanglement assistance. The quantum capacity region $\opQ_{\text{EA*}}(\channel)$ is defined in a similar manner as before. %
\end{definition}

In the following remark, we discuss the relation between the classical and quantum formulations above.
In many communication models in the literature, it does not matter whether the messages are chosen by the sender Alice, or given to her by an external source.
However, in the quantum model, there is a fundamental distinction between the general task of sub-space transmission
and remote state preparation, as we explain below.

\begin{remark}
In the classical code in Definition~\ref{def:EAcapacity}, if entanglement assistance is present, then Bob decodes the composite message $\bar{m}=(m,m')$. Hence, the overall transmission rate with entanglement assistance is $R_{EA}=R+R'$.
In the quantum code in Definition~\ref{def:QEAcapacity}, $M$ and $\bar{M}$ are two independent systems of dimensions $2^{nQ}$ and $2^{n(Q+Q')}$. %
Hence, the overall quantum rate with entanglement assistance is $Q_{EA}=Q+Q'$.
In some applications of quantum error correction, Alice receives the system $M$ from another source, and does not prepare it herself. 
While Alice can perform any operation on this system, she does not necessarily know its state in this case. 
Due to the no-cloning theorem, Alice cannot duplicate a general state of $M$ either. 
Thus, 
our definition of quantum transmission with unreliable entanglement %
describes a more restricted problem. 
\end{remark}

\subsection{Related Work}
\label{subsec:Previous}
We briefly review known results without %
assistance and with reliable entanglement assistance.
We denote the corresponding classical and quantum capacities with reliable entanglement assistance by $C_{\text{EA}}(\channel)$ and
  $Q_{\text{EA}}(\channel)$, and without assistance by $C(\channel)$ and
 $Q(\channel)$, respectively. %

Define the following information measures: The channel Holevo information
\begin{align}
&\chi(\channel)\equiv   \max_{p_X(x), \ket{\psi_A^x} } I(X;B)_\omega  \,,
\nonumber\\
&\omega_{XB}\equiv \sum_{x\in\Xset} p_X(x) \kb{x}\otimes \channel( \psi_A^x ) \,,\; \abs{\Xset}\leq \abs{\Hset_A}^2 \,,
\label{eq:HolevoChan}
\end{align}
and the channel coherent information
\begin{align}
&I_c(\channel)\equiv \max_{\ket{\phi_{A_1 A}} } I(A_1\rangle B)_\omega \,,
\nonumber\\
& \omega_{A_1 B}\equiv  (\text{id}\otimes\channel)( \kb{ \phi_{A_1 A} }) \,,\; \abs{\Hset_{A_1}}\leq \abs{\Hset_A} \,.
\label{eq:CoherentChan}
\end{align}
Observe that the Holevo information is maximized over ensembles of pure states, while the quantum capacity is maximized over entangled states. 
We will see the implications of those properties in the results section. 
The classical capacity theorem and the quantum capacity theorem are given below.
\begin{theorem} %
\label{theo:CclNoEA}
The \emph{classical capacity} of a quantum channel $\channel_{A\rightarrow B}$ without assistance is given by \cite{Holevo:98p,SchumacherWestmoreland:97p%
}
\begin{align}
C(\channel)= \lim_{k\rightarrow \infty} \frac{1}{k} \chi \left( \channel^{\otimes k} \right) \,.
\end{align}
\end{theorem}

\begin{theorem} 
\label{theo:CNoSI}
The \emph{quantum capacity} of a quantum channel $\channel_{A\rightarrow B}$  without assistance is given by \cite{BarnumNielsenSchumacher:98p,Loyd:97p,Shor:02l,Devetak:05p%
} 
\begin{align}
Q(\channel)= \lim_{k\rightarrow\infty}\frac{1}{k} I_c(\channel^{\otimes k}) \,.
\label{eq:CqNosi}
\end{align}
\end{theorem}

Next, consider communication with reliable entanglement assistance. The entanglement-assisted capacity formula turns out to be the quantum analog of Shannon's classical formula \cite{BennettShorSmolin:99p,BennettShorSmolin:02p}. Define
\begin{align}
&I(\channel)= \max_{\ket{\phi_{A_1 A}} } I(A_1; B)_\omega \,,
\label{eq:inCea}
\nonumber\\
& \omega_{A_1 B}\equiv  (\text{id}\otimes\channel)( \kb{ \phi_{A_1 A} }) \,,\; \abs{\Hset_{A_1}}\leq \abs{\Hset_A} \,.
\end{align}
\begin{theorem} 
\label{theo:CeaNoSI}
The classical capacity and the quantum capacity of a quantum channel $\channel_{A\to B}$ with reliable entanglement assistance are given by \cite{BennettShorSmolin:99p,BennettShorSmolin:02p}
\begin{align}
C_{\text{EA}}(\channel) &= I(\channel) \,, %
\\
Q_{\text{EA}}(\channel) &= \frac{1}{2}I(\channel) \,.
\label{eq:QeaNoMsk}
\end{align}
\end{theorem}
The classical capacity and the quantum capacity have different units, \ie $C_{\text{EA}}(\channel)$ is measured in classical information bits per channel use, whereas $Q_{\text{EA}}(\channel)$ in information \emph{qubits} per channel use.
Nonetheless,  the capacity values satisfy $Q_{\text{EA}}(\channel)=\frac{1}{2} C_{\text{EA}}(\channel) $, given reliable entanglement assistance. This relation can be inferred from the fundamental single-unit protocols.
Specifically, super-dense coding \cite{BennetWiesner:92p} is a well known communication protocol whereby two classical bits are transmitted using a single use of a noiseless qubit channel and a maximally entangled pair. In the other direction,  by employing the teleportation protocol \cite{BennettBrassardJozsaPeres:93p}, qubits can be sent %
at half the rate of classical bits given entanglement resources.

\section{Results}
We establish a regularized characterization for the capacity region with unreliable entanglement assistance, for the transmission of either classical information or quantum information.

\subsection{Classical Communication}
Let $\channel_{ A\rightarrow B}$ be a given channel, and
define
\begin{align}
\mathcal{R}_\text{EA*}(\channel)
&=
\bigcup_{ p_X \,,\; \varphi_{A_0 A_1} \,,\; \Fset^{(x)} 
}
\left\{ \begin{array}{rl}
  (R,R') \,:\;
	R \leq& I(X;B)_\omega  \\
  R'   \leq& I(A_1;B|  X)_\omega
	\end{array}
\right\}
\label{eq:calRClea}
\intertext{with}
\omega_{XA_1  A}
&=
\sum_{x\in\Xset} p_X(x) \kb{x}\otimes ( \text{id}\otimes \Fset^{(x)}_{A_0\to A})
(\varphi_{ A_1 A_0}) \,,  \label{eq:calRClea1}\\
\omega_{XA_1  B}
&=
(\text{id}\otimes\channel_{A\rightarrow B})(\omega_{XA_1  A})
\,.
\label{eq:calRClea2}
\end{align}
Intuitively, the classical variable $X$  is associated with the classical message $m$, which Bob decodes whether there is entanglement assistance or not. The reference system $A_0$ can be thought of as Alice's share of the entanglement resources. Since the resources are pre-shared before communication takes place, the entangled state $\varphi$ is non-correlated with the messages.
 Alice encodes the message $m'$  using the encoding operator $ \Fset^{(x)} $. 
Before we state the capacity theorem, we give the following lemma. The property below simplifies the computation of the above region and the achievability proof as well. %
\begin{lemma}
\label{lemm:pureCea}
The union in (\ref{eq:calRClea}) is exhausted by pure states $\ket{\phi_{A_0 A_1}}$ and with the cardinality of %
$\abs{\Xset}\leq \abs{\Hset_A}^2+1$. 
\end{lemma}
The restriction to pure states %
is based on %
 state purification, %
while %
the alphabet bound follows 
from the Fenchel-Eggleston-Carath\'eodory lemma \cite{Eggleston:66p}, using similar arguments that Yard \etal \cite{YardHaydenDevetak:08p} use. 
 The details are given in 
Appendix~\ref{app:pureCea}. Our main result on classical communication with unreliable entanglement assistance is stated below.
\begin{theorem}
\label{theo:ClEA}
The classical capacity region of a quantum channel $\channel_{A\to B}$ with unreliable entanglement assistance satisfies
\begin{align}
\mathcal{C}_\text{EA*}(\channel)=& \bigcup_{n=1}^\infty \frac{1}{n} \mathcal{R}_\text{EA*}(\channel^{\otimes n}) \,.
\end{align}
\end{theorem}
The proof of Theorem~\ref{theo:ClEA} is given in Appendix~\ref{app:ClEA}. 
In general, there is a tradeoff between the rates $R$ and $R'$, and we cannot necessarily achieve the maximum rate for both of them simultaneously. Intuitively, the excess rate $R'$ that is provided by entanglement assistance depends on the level of entanglement between the ancilla $A_1$ and the channel input $A$, or equivalently, on how entanglement-breaking the encoding map is. We give a more precise explanation below. 

In the region formula in (\ref{eq:calRClea}), we have a union over the probability distributions $p_X$, states  $\varphi_{A_0 A_1}$, and collections of mappings $\{\Fset_{A_0\to A}^{(x)}\}_{x\in\Xset}$. The boundary of this region is attained by optimizing over these objects.  Observe that in order for $R'$ to achieve the entanglement-assisted capacity, we may set $\varphi_{A_0 A_1}$ as the entangled state that attains the maximum in (\ref{eq:inCea}), and take $\Fset_{A_0\to A}^{(x)}$ to be the identity map. Since the output has no correlation with $X$, this assignment achieves the rate pair $(R,R')=(0,C_{\text{EA}}(\channel))$ (\cf (\ref{eq:inCea}) and (\ref{eq:calRClea})).

To maximize the unassisted rate, set an encoding channel $\Fset_{A_0\to A}^{(x)}$ that outputs the pure state $|\psi_A^x\rangle$ that is optimal in (\ref{eq:HolevoChan}), \ie
\begin{align}
\Fset^{(x)}(\varphi_{A_1 A_0})=\varphi_{A_1}\otimes \psi_A^x \,.
\end{align}
Such an assignment achieves $(R,R')=(\chi(\channel),0)$ (\cf (\ref{eq:HolevoChan}) and (\ref{eq:calRClea})).
In other words, the Holevo information is achieved for an entanglement-breaking encoder.

\begin{remark}
\label{remark:Broadcast_Relation}
Our model has a deep relation to the quantum broadcast channel. 
We point out a heuristic connection, while the precise formulation is delegated to the discussion section
(see Subsection~\ref{subsection:One_Side_Entnaglement}).
The characterization of the classical capacity region %
in Theorem~\ref{theo:ClEA} clearly  resembles the classical SPC\begin{footnote}{A superposition code should not be confused with quantum superposition. The two notions are unrelated.}\end{footnote} region  of the broadcast channel without assistance  \cite{YardHaydenDevetak:11p} (see Theorem 2 therein).
The difference is that here the encoding involves quantum operations.
Nevertheless, we can portray a similar metaphorical image:
Let %
 $m$ be an index over $2^{nR}$ clouds. Recall that the region formula in (\ref{eq:calRClea}) involves an ancillary state $\varphi_{A_0 A_1}$ and a collection of encoding mappings $\{ \Fset^{(x)}
 \}_{x\in\Xset}$.
Each cloud center is associated with a classical codeword $x^n(m)$, and
at the center of each cloud there is the state $\otimes_{i=1}^n \Fset^{(x_i(m))}(\varphi)$.  Applying random Heisenberg-Weyl operators that 
encode the message $m'$  takes us from the cloud center to a satellite on the cloud that depends on both $m$ and $m'$. The channel input is the satellite state.
See Figure~\ref{fig:Superposition}.
Bob decodes in two steps. First, Bob recovers the cloud, \ie he estimates $m$.
If the entanglement assistance is absent, then Bob quits after the first step.
Otherwise, if Bob has entanglement assistance, then he continues to decode the satellite $m'$.
\end{remark}

\begin{remark}
The results are very surprising when compared with classical systems.
SPC is a sophisticated network technique that is mainly used for multi-user channels or other network configurations in order to achieve a tradeoff between different users or resources. 
For example, consider a transmitter $X$ that communicates two messages over a classical broadcast channel with two receivers, $Y_1$ and $Y_2$, while each message is intended for a different receiver.
Of course, it is not necessarily possible to maximize the transmission rates for both users simultaneously. 
However,
if the outputs are identical, \ie $Y_1=Y_2\equiv Y$, then
%
the capacity region is given by the set of rate pairs 
$(R_1,R_2)$ such that $R_1+R_2\leq C_1$, where $C_1
$ the Shannon capacity of the point-to-point channel $P_{Y |X}$.
In such a simple case, the elaborate SPC scheme is not needed, and the capacity region can be achieved by 
the much simpler \emph{time division} approach.
That is, a concatenation of two single-user codes is optimal.

In our  model, we consider a point-to-point quantum channel $\channel_{A\to B}$.
Thereby, it may appear at a first glance as if time division should be optimal, regardless of whether the channel is noisy or not.
In the example below, we show that time division can be sub-optimal.
To this end, we apply Theorem~\ref{theo:ClEA} with a quantum state $\ket{\phi_{A_0 A_1}}$ that is formed by the superposition of a product state and a maximally entangled Bell state.
Hence, despite the simplicity of this single-user point-to-point model, one can outperform time division by inserting a quantum superposition state into an SPC scheme.

\end{remark}

\begin{remark}
Returning to our metaphor of
 travelers  on a sea journey, we now consider a division plan.
In this scenario, $N$ travelers are divided between two ships: a light ship without lifeboats and a heavy ship with lifeboats that can accommodate all passengers. The light ship has a maximum excess speed of $R_{\text{light}}'=V(N,0)$, but no guaranteed speed since it does not have any lifeboats ($L=0$). The heavy ship, on the other hand, has a high guaranteed speed of $R_{\text{heavy}}=v_0$, but a low excess speed of $R_{\text{heavy}}'$. Although the heavy ship is less efficient, it is more reliable. By dividing the passengers between the two ships, we can achieve an average speed pair of $(R,R')=(1-\lambda)(R_{\text{light}},R_{\text{light}}')+\lambda (R_{\text{heavy}},R_{\text{heavy}}')$, where $\lambda$ represents the fraction of passengers on the heavy ship normalized by the total number of passengers. 
Figuratively,
our results show that if the journey is subject to 
a quantum evolution, then we may outperform the division plan by
allowing travelers to be in a quantum superposition state between the two ships.
\end{remark}
 
To demonstrate our results, we give an example.
\begin{example}
\label{example:depolCl}
Consider the qubit depolarizing channel
\begin{align}
\channel(\rho)=(1-\eps)\rho+\eps \frac{\identity}{2} \,,
\end{align}
with $\eps\in [0,1]$. The classical capacity without assistance is given by $C(\channel)=%
1-H_2\left( \frac{\eps}{2} \right)$, and it is achieved with a symmetric distribution over the ensemble $\{ |0\rangle, |1\rangle \}$, where 
$H_2(t)\equiv -t\log(t)-(1-t)\log(1-t)$ is the binary entropy function \cite{King:03p}.
On the other hand, the classical capacity with reliable entanglement assistance is given by $C_{\text{EA}}(\channel)=2-H\left(1- \frac{3\eps}{4},\frac{\eps}{4},\frac{\eps}{4},\frac{\eps}{4} \right)$, and it is achieved with a maximally entangled input state \cite{BennettShorSmolin:99p}.

A natural compromise is to mix the strategies above. %
Let $Z$ be an independent random bit that chooses between the strategies, %
where $Z\sim \text{Bernoulli}(\lambda)$ for a given $\lambda\in [0,1]$.
That is, we define $\Fset^{(x,z)}$ by $\Fset^{(x,0)}(\rho_A)=\psi_A^x$ and $\Fset^{(x,1)}=\text{id}$. Plugging $\tilde{X}\equiv (X,Z)$, we obtain the time-division %
achievable region,
\begin{align}
\mathcal{R}_\text{EA*}(\channel)
&\supseteq
\bigcup_{0\leq \lambda\leq 1}
\left\{ \begin{array}{rl}
  (R,R') \,:\;
	R \leq& (1-\lambda)\, C(\channel)   \\
  R'   \leq& \lambda C_{\text{EA}}(\channel)
	\end{array}
\right\} \,.
\end{align}

Next, we numerically compute an achievable region that outperforms the time-division bound.
 Instead of using a classical mixture of the strategies, we use quantum superposition. Define a non-normalized vector,
\begin{align}
|u_\beta \rangle \equiv \sqrt{1-\beta} \ket{0}\otimes \ket{0}+\sqrt{\beta} \ket{\Phi} \,.
\label{eq:ubeta}
\end{align}
Then, set
\begin{align}
\ket{\phi_{A_0 A_1}}&\equiv \frac{1}{\norm{u_\beta}} \ket{u_\beta} \,, \\
p_X &= \left( \frac{1}{2},\frac{1}{2} \right) \,,\\
\Fset^{(x)}(\rho)&\equiv \mathsf{X}^x \rho \mathsf{X}^x \,,
\end{align}
where $\mathsf{X}$ is the bitflip Pauli operator.
Observe that for $\beta=0$, the input state is $\Fset^{(x)}(\ketbra{0})=\ketbra{x}$, which achieves the classical capacity without assistance. On the other hand, for $\beta=1$, the parameter $x$ chooses one of two Bell states.

Figure~\ref{fig:clEAur} depicts the resulting region for a depolarization probability of $\eps=\frac{1}{2}$.
The triangular region below the dashed red line is the time-division bound, which is obtained by a classical mixture, %
whereas  the solid blue line indicates the achievable region corresponding to the superposition state $\frac{\ket{u_\beta}}{\norm{u_\beta}}$, as in (\ref{eq:ubeta}).

\end{example}
\begin{figure}[tb]
\center
\includegraphics[scale=0.09,trim={6cm 0 0 0},clip]
{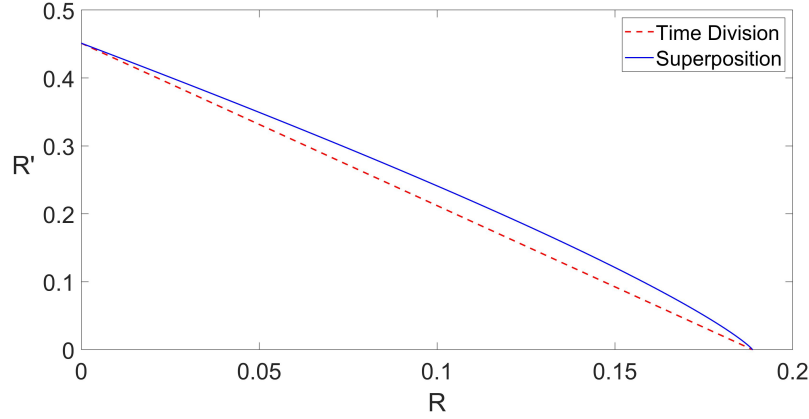} %
\caption{Achievable classical rate regions for the depolarizing channel with unreliable entanglement assistance, for a depolarization probability $\eps=\frac{1}{2}$.
}
\label{fig:clEAur}
\end{figure}

For a noiseless qubit channel, time division is optimal.
\begin{corollary}
The classical capacity region of a noiseless qubit channel  with unreliable entanglement assistance is given by the time-division region, \ie
\begin{align}
\mathcal{C}_\text{EA*}(\text{id})
&=
\bigcup_{ 0\leq \lambda\leq 1} 
\left\{ \begin{array}{rl}
  (R,R') \,:\;
	R \leq & 1-\lambda  \\
  R'   \leq& 2\lambda
	\end{array}
\right\} \,.
\label{eq:calRClea0}
\end{align}
\end{corollary}
\begin{proof}
As explained in the introduction, to achieve the rate pair $(R,R')=(1-\lambda,2\lambda)$, Alice and Bob simply perform super-dense coding repeatedly, over a fraction of $\lambda$ of the block, and communicate over the  $0$-$1$ basis in the remaining part. 
To show the converse part, let $(R,R')\in \frac{1}{n}\mathcal{R}_\text{EA*}(\channel^{\otimes n})$
(see (\ref{eq:calRClea2})), hence 
\begin{align}
R&\leq \frac{1}{n}I(X;B^n)_\omega
=\frac{1}{n} [H(B^n)_\omega- H(B^n|X)_\omega]
\nonumber\\
&\leq 1-\frac{1}{n} H(B^n|X)_\omega
\intertext{and}
R'&\leq \frac{1}{n}I(A_1;B^n|  X)_\omega
\nonumber\\
&\leq \frac{1}{n}\cdot 2 H(B^n|X)_\omega \,.
\end{align}
The last inequality holds since the quantum conditional entropy satisfies  $|H(B|$ $A)_\rho|$ $\leq H(B)_\rho$ in general, hence
$H(B^n| A_1  X)_\omega=\sum_{x} p_X(x) H(B^n|A_1,X=x)_\omega \geq  -\sum_{x} p_X(x) H(B^n|X=x)_\omega =-H(B^n|X)_\omega
$.
The converse part for the corollary follows, as we define $\lambda\equiv \frac{1}{n} H(B^n|X)_\omega$. 
\end{proof}

\subsection{Quantum Communication}
Consider quantum communication over  $\channel_{ A\rightarrow B}$ with unreliable entanglement assistance.
Define
\begin{align}
\inQ_\text{EA*}(\channel)=
\bigcup_{  \varphi_{A_1 A_2 A} }
\left\{ \begin{array}{l}
  (Q,Q') \,:\;\\
	Q \leq \min\{ I(A_1\rangle B)_\omega \,,\; H(A_1|  A_2)_\omega \} \,,\\
  Q+Q'   \leq \frac{1}{2} I(A_2;B)_\omega
	\end{array}
\right\}
\label{eq:calRQea}
\end{align}
with
\begin{align}
\omega_{A_1 A_2  B}&=(\text{id}\otimes\channel_{A\rightarrow B})(\varphi_{A_1 A_2  A})
\,.
\end{align}
\begin{theorem}
\label{theo:qEA}
The quantum capacity region of a quantum channel $\channel_{A\to B}$ with unreliable entanglement assistance satisfies
\begin{align}
\mathcal{Q}_\text{EA*}(\channel)=& \bigcup_{n=1}^\infty \frac{1}{n} \inQ_\text{EA*}(\channel^{\otimes n}) \,.
\end{align}
\end{theorem}
The proof of Theorem~\ref{theo:qEA} is given in Appendix~\ref{app:qEA}.

\section{Summary and Discussion}
\label{sec:summ}
We summarize our results and compare  the techniques in our work and in previous works.
We consider communication %
over a quantum channel $\channel_{A\to B}$, where Alice and Bob are provided with \emph{unreliable} entanglement resources. %
Suppose that  Alice wishes to send two messages, at rates $R$ and $R'$. She encodes both messages using her share of the entanglement resources, as she does not know whether Bob will have access to the entangled resources.
Bob has two decoding procedures. If the entanglement assistance has failed to reach Bob's location, he performs a decoding operation to recover the first message alone. Hence, the communication system operates on a rate $R$.
Whereas if Bob has entanglement assistance, he decodes both messages, hence the overall transmission rate is $R+R'$.
In other words, $R$ is a \emph{guaranteed rate}, and $R'$ is the \emph{excess rate} of information that entanglement assistance provides.
The communication setting is illustrated in Figure~\ref{fig:EAsiCode}, in which the resource uncertainty is represented by the unknown position of  a switch.

We define the capacity region as the set of all rate pairs $(R,R')$ that can be achieved %
with asymptotically vanishing decoding errors.
The characterization of the corner  points $(R,0)$ and $(0,R')$ already follows from previous results in the literature (see Subsection~\ref{subsec:Previous}).
However, our
interest goes beyond those cases, 
as we focus on the tradeoff, and not the extreme points. That is,
we  characterize the entire capacity region.

In the transmission of quantum information, Alice %
chooses a product state $\theta_M\otimes \xi_{\bar{M}}$ over Hilbert spaces of dimension 
$|\Hset_M|=2^{nQ}$ and $|\Hset_{\bar{M}}|=2^{n(Q+Q')}$. %
Alice encodes the input state by applying the encoding channel 
$\Fset_{G_A M \bar{M} \to A^n}$ to $M$, $\bar{M}$, and to her share of the entangled state $\Psi_{G_{A},G_{B}}$,  and transmits the system $A^n$ over $n$ channel uses of 
$\channel_{A\to B}$. Bob receives the channel output systems $B^n$.
 If the entanglement assistance is present, then he %
applies the decoding channel
 $\Dset%
$ to the joint output $B^n G_B$ in order to recover $\xi_{\bar{M}}$. 
Otherwise, if entanglement assistance is absent, then he performs $\Dset^*$ on $B^n$ in order to recover $\theta_M$.

 We have established a regularized characterization for the classical and quantum capacity regions. %
The communication design makes a compromise. 
We have seen that the unassisted rate is high when the channel input is in a pure state.
Such an assignment achieves $(R,R')=(\chi(\channel),0)$, where $\chi(\channel)$ is the Holevo information of the channel. %
On the other hand, high excess rates are achieved when the encoder preserves the entanglement between the input system and the ancilla (see (\ref{eq:calRClea})). %
Such an encoding operation achieves the rate pair $(R,R')=(0,C_{\text{EA}}(\channel))$, where $C_{\text{EA}}(\channel)$ is the entanglement-assisted capacity. %
Time division between these coding strategies achieves the rate pairs $(R,R')=((1-\lambda)\chi(\channel),\lambda C_{\text{EA}}(\channel))$, for $0\leq\lambda\leq 1$.

In the simple scenario of classical communication over a noiseless qubit channel, we have shown that the optimal strategy is to perform time division between super-dense coding and unassisted transmission over the  $0$-$1$ basis (see Figure~\ref{fig:DenseUncertain}). 
For the noiseless qubit channel,
 the classical capacity region  with unreliable entanglement assistance is thus $\mathcal{C}_{\text{EA}*}(\channel)=$   $\bigcup_{0\leq\lambda\leq 1}\{(R,R')\,:\; R\leq 1-\lambda \,,\; R'\leq 2\lambda \}$.
 Nevertheless, we established that time division is not optimal in general, even for a simple noisy quantum channel.
Specifically,  time division is strictly suboptimal for the depolarizing channel (see Figure~\ref{fig:clEAur}).

Thereby, more advanced coding techniques are necessary to obtain the full capacity region.
Indeed,
our characterization in Theorem~\ref{theo:ClEA} and the coding method in the achievability proof are much more sophisticated than time division.
The analysis is based on a novel method that is inspired by the classical network technique of superposition coding (SPC) \cite{CoverThomas:06b}.
Surprisingly, this network technique yields an advantage even for a simple point-to-point quantum channel. This advantage is obtained by exploiting  quantum superposition. 
That is, we combine superposition coding with superposition states.

Next, we discuss capacity computation, 
the side-information interpretation, and the consequences on the quantum broadcast channel with one-sided entanglement assistance.

\subsection{Computing Channel Capacities}
For communication system design nowadays, %
it is crucial to evaluate the current performance and how close it is to the optimum %
\cite{AroaraSinghRandahawa:19p,CostelloForney:07c}. %
Classical commercial systems today already employ sophisticated error correction codes with near-Shannon limit performance \cite{NisiotiThomos:20a,RichardsonKudekar:18p}.
At the time of writing, a realization of a full-scale quantum communication system that approaches the Shannon-theoretic limits does not exist, and we can only hope that future systems of quantum communication will reach this level of maturity.
Given a specific quantum channel, e.g. an optical fiber channel with specific parameters, a practitioner is usually interested in computing the channel capacity as a number. %
	For such practical purposes, %
 a regularized characterization as in Theorems \ref{theo:CclNoEA}-\ref{theo:CNoSI}, \ref{theo:ClEA}, and \ref{theo:qEA}  is not necessarily a problem (see a further explanation in Remark 7 by the authors \cite{PeregDeppeBoche:21p}). 
Yet, in Shannon theory, it is generally considered desirable to establish a single-letter computable capacity formula \cite{Korner:87b,Pereg:21p}. Beyond computability, the  disadvantage of a regularized multi-letter formula of the form 
$%
\lim_{n\rightarrow\infty}\frac{1}{n}\mathsf{F}(\channel^{\otimes n}) %
$, %
 is that such characterization is not unique (see \cite[Section 13.1.3]{Wilde:17b}). 

Under practical  encoding constraints \cite{Pereg:21p},
regularized capacity results yield computable formulas. %
Encoding constraints   are particularly relevant
when the transmitter has access to a cluster of multiple small or moderate-size quantum computers without interaction between them, and also in nearest-neighbor qubit architectures \cite{StassiCirioNori:20p,LMRDFLWM:17p}. %
Consider classical communication without assistance, as in Theorem \ref{theo:CclNoEA}, and
suppose that the encoder's quantum systems $A^n$ are partitioned into sub-blocks of a small size $b$, such that the input state has the form
$%
\rho_{A^n}=  \rho_{A_1^b} \otimes \rho_{A_{b+1}^{2b}} \otimes \cdots  \otimes \rho_{A_{n-b+1}^{n}} %
$. %
As recently observed  \cite{Pereg:21p}, the
capacity of a quantum channel $\channel_{A\to B}$ under an encoding constraint $b>0$, is given by 
\begin{align}
C(\channel,b)=\frac{1}{b}\chi(\channel^{\otimes b}) \,.
\label{eq:Cb}
\end{align}
This %
formula  is computable, since $b>0$ is assumed to be a small constant.
This trivial observation and its consequences can be extended to other models as well.

Another shortcoming of our results is that we do not have a bound on the dimension of the ancillas  $A_1$ and $A_2$.
One could always compute an achievable region by simply choosing the dimensions of $A_1$ and $A_2$. However, the optimal rates cannot be computed with absolute precision in general.
A similar difficulty  appears in other quantum models such as  broadcast communication \cite[Section VIII]{DupuisHaydenLi:10p}, %
the wiretap channel  \cite[Remark 5]{QiSharmaWilde:18p}, squashed entanglement \cite[Section 1]{LiWinter:14p}, and state-dependent channels \cite{Dupuis:09c} %
\cite[Section V]{PeregDeppeBoche:21p}.


\subsection{Side Information Interpretation}
We mentioned that 
our coding approach can be interpreted as a quantum version of SPC (see Remark~\ref{remark:Broadcast_Relation}).
Consider the second decoding step for the message $m'$.
As the message $m$ has already been estimated, we can think of $x^n(m)$ as side information for this decoding operation.
Thus, it is not surprising that the bound on the excess rate $R'$ in (\ref{eq:calRClea}) has a similar form as in the capacity formula for a quantum channel with classical side information at the encoder and the decoder (see \cite[Corollary 12]{Pereg:19a}).

For the quantum capacity region, we point out a connection to quantum side information. %
A quantum state-dependent channel $(\Pset_{S A\rightarrow B} ,|\theta_{S S_0}\rangle)$ is defined by a   linear, completely positive, trace preserving map 
$%
\Pset_{S A\rightarrow B}  %
$ %
and a fixed quantum state $|\theta_{S S_0}\rangle$ \cite{Dupuis:09c}. We refer to the system $S$ as a quantum channel state.
Given quantum side information, the encoder has access to the system $S_0$, which is entangled with the channel state system $S$.
This model can be interpreted as if the channel is entangled with the systems $S$ and $S_0$.
The quantum capacity with quantum side information and no assistance is given by the regularization of the following formula (see \cite[Theorem 11]{PeregDeppeBoche:21p}),
\begin{align}
&L(\Pset)=
\sup_{ \varphi_{ A_1 S A  } \,:\: \varphi_{S}=\theta_{S}}
	 \min\{ I(A_1\rangle B)_\omega \,,\; H(A_1|S )_\varphi \}  \,,
\end{align} 
with $\omega_{A_1  B }=\Pset_{S A\rightarrow B}(\varphi_{ A_1 S A})$. 
Thus, we interpret the guaranteed rate $Q$ in (\ref{eq:calRQea}) as the quantum coding rate, given access to %
a channel state system $A_2$.

\subsection{The Broadcast Channel with One-Sided Assistance}
\label{subsection:One_Side_Entnaglement}
Beyond the heuristic connection, 
the mathematical formulation of our problem is close to that of a broadcast channel with one-sided assistance.
Let $\channel^{\text{broadcast}}_{A\to B_1 B_2}$ be a quantum broadcast channel with two receivers, Bob 1 and Bob 2.
Suppose that Alice wishes to send a common message $m_0\in [1:2^{nR_0}]$ to both users and a dedicated message $m_1\in [1:2^{nR_1}]$ to the first user alone.
That is, Bob 1 decodes both $m_0$ and $m_1$, while Bob 2 is only required to decode $m_0$.  This model is referred to as the broadcast channel with degraded message sets \cite{YardHaydenDevetak:11p}.
Now, assume that Alice and Bob 1 share reliable entanglement resources $\Psi_{G_A,G_{B_1}}$, while Bob 2 has no resources at all.

 The error criterion is  the probability that at least one of the receivers decodes erroneously. However, it is sufficient to consider each receiver separately, since the coding performance %
depends on the broadcast channel $\channel^{\text{broadcast}}_{A\to B_1 B_2}$ only through the marginals $\channel^{(1)}_{A\to B_1}$ and $\channel^{(2)}_{A\to B_2}$ \cite{PeregDeppeBoche:21p2}, %
\begin{align}
\channel^{(1)}(\rho_A)&\equiv \trace_{B_2}\left(\channel^{\text{broadcast}}(\rho_A) \right) \,, \\
\channel^{(2)}(\rho_A)&\equiv \trace_{B_1}\left(\channel^{\text{broadcast}}(\rho_A) \right) \,.
\end{align}
Hence, achievable rate pairs $(R_0,R_1)$ can be defined in terms of the following error probabilities,
\begin{align}
&P_{e1|  m_0,m_1}^{(n)}(\Fset,\Psi,\Dset^{(1)})= 1- %
\trace\big[ D^{(1)}_{m_0,m_1}
(\channel^{(1)\otimes n}_{A\rightarrow B_1}\otimes\text{id})
(\Fset^{m_0,m_1}\otimes\text{id}) (\Psi_{G_A,G_{B_1}})
 \big] %
\intertext{for Bob 1, and}
&P_{e2|  m_0,m_1}^{(n)}(\Fset,\Psi,\Dset^{(2)})= %
1-
\trace\Big[ D_{m_0}^{(2)}\,
\channel^{(2)\otimes n}_{A\rightarrow B_2}\,
\Fset^{m_0,m_1}(\Psi_{G_A})  \Big] 
\end{align}
for Bob 2, where $\Fset_{G_A\to A^n}^{m_0,m_1}$ is the encoding map, while $\Dset^{(1)}_{G_{B_1}B_1^n}=\{ D^{(1)}_{m_0,m_1} \}$ and $\Dset_{B_2^n}^{(2)}=\{ D^{(2)}_{m_0} \}$ are the decoding maps of Bob 1 and Bob 2, respectively. 

Observe that the error definitions above are analogous to those of the classical capacity region with unreliable entanglement assistance in Definition~\ref{def:EAcapacity}, where $m$ and $m'$ are replaced by $m_0$ and $m_1$, respectively. Although,  the error probabilities for the broadcast channel depend on two different channels, $\channel^{(1)}$ and $\channel^{(2)}$. 
The same methods as we used in this paper show that the classical capacity region of the quantum broadcast channel with one-sided entanglement assistance is given by the regularization of the following formula,
\begin{align}
\mathcal{R}_2(\channel^{\text{broadcast}})
&=
\bigcup_{ p_X \,,\; \varphi_{A_0 A_1} \,,\; \Fset^{(x)} 
}
\left\{ \begin{array}{rl}
  (R_0,R_1) \,:\;
	R_0 \leq&  I(X;B_2)_\omega  \\
  R_1 \leq&  I(A_1;B_1|  X)_\omega\\
	R_0+R_1 \leq&  I(XA_1;B_1)_\omega
	\end{array}
\right\}
\label{eq:calRCleaB}
\intertext{with}
\omega_{XA_1  A}
&=
\sum_{x\in\Xset} p_X(x) \kb{x}\otimes ( \text{id}\otimes \Fset^{(x)}_{A_0\to A})
(\varphi_{ A_1 A_0}) \,,  \label{eq:calRClea1B}\\
\omega_{XA_1  B_1 B_2}
&=
(\text{id}\otimes\channel^{\text{broadcast}}_{A\rightarrow B_1 B_2})(\rho_{XA_1  A})
\,.
\label{eq:calRClea2B}
\end{align}

It is now natural to wonder whether this similarity extends to the quantum capacity.
However, in the transmission of quantum information, Bob 1 and Bob 2 cannot recover a common state due to the no-cloning theorem. That is, quantum communication with  degraded message sets is not well defined (see \cite[Section III.C]{PeregDeppeBoche:21p2}).
Yet, the techniques of %
Dupuis \etal \cite{Dupuis:10z,DupuisHaydenLi:10p} for the quantum broadcast channel with dedicated messages were useful in our proof  %
in Appendix~\ref{app:qEA}, for the quantum capacity theorem with unreliable entanglement assistance.

\section*{Acknowledgments}
The authors wish to thank Christoph Becher (Universit\"at des Saarlandes) for useful discussions,
 Elizabeth S\"oder for the English review, and Mohammad J. Salariseddigh for his technical support in preparing the figures.

U. Pereg was supported by the Israel VATAT Junior Faculty Program for Quantum Science and Technology through Grant 86636903, and the Chaya Career Advancement Chair, Grant 8776026.
U. Pereg and H. Boche were also supported 
by the Deutsche Forschungsgemeinschaft (DFG, German Research Foundation) under Germany's Excellence Strategy – EXC-2111 – 390814868. 
In addition,
U. Pereg, C. Deppe, and H. Boche were supported by
the German Federal
Ministry of Education and Research (BMBF) through Grants
16KISQ028 (Pereg, Deppe) and 16KISQ020 %
(Boche). 
This work of H. Boche was supported in part by the BMBF within the national initiative for
``Post Shannon Communication (NewCom)" under Grant 16KIS1003K, and in
part by the DFG within the Gottfried Wilhelm
Leibniz Prize under Grant BO 1734/20-1 and within Germany's Excellence
Strategy EXC-2092 – 390781972.

\begin{appendices} %
{
\section{Information-Theoretic Tools}
In this section, we give the basic information-theoretic tools that will be used in the achievability proofs later on.

\subsection{Quantum Packing Lemma}
\label{app:packing}
To prove achievability for the classical capacity theorem, we will use the quantum packing lemma. Standard method-of-types concepts are defined as usual \cite{Wilde:17b} \cite{Pereg:21p}.  
We briefly introduce the notation and basic properties while the detailed definitions can be found in the references \cite%
{Pereg:21p}.
In particular, given a density operator $\rho=\sum_x p_X(x)\kb{x}$ on the Hilbert space $\Hset_A$, we let
$\tset(p_X)$ denote the $\delta$-typical set that is associated with $p_X$, and
 $\Pi_{A^n}^{\delta}(\rho)$ the projector onto the corresponding subspace.  
The following inequalities follow from well-known properties of $\delta$-typical sets \cite{NielsenChuang:02b}, %
\begin{align}
\trace( \Pi^\delta(\rho) \rho^{\otimes n} )&\geq 1-\eps  \label{eq:UnitT} \\
 2^{-n(H(\rho)+c\delta)} \Pi^\delta(\rho) &\preceq \,\Pi^\delta(\rho) \,\rho^{\otimes n}\, \Pi^\delta(\rho) \,
\preceq 2^{-n(H(\rho)-c\delta)}
\label{eq:rhonProjIneq}
\\
\trace( \Pi^\delta(\rho))&\leq 2^{n(H(\rho)+c\delta)} \label{eq:Pidim}
\end{align}
 where $c>0$ is a constant, and $\eps>0$ tends to zero as $\delta\to 0$.
Furthermore, for $\sigma_B=\sum_x p_X(x)\rho_B^x$, %
let $\Pi_{B^n}^{\delta}(\sigma_B|  x^n)$ denote the projector corresponding to the conditional $\delta$-typical set %
given the sequence $x^n$.
Similarly \cite{Wilde:17b}, %
\begin{align}
\trace( \Pi^\delta(\sigma_B|  x^n) \rho_{B^n}^{x^ n} )&\geq 1-\eps'  \label{eq:UnitTCond} \\
 2^{-n(H(B|  X')_\sigma+c'\delta)} \Pi^\delta(\sigma_B|  x^n) &\preceq \,\Pi^\delta(\sigma_B|  x^n) \,\rho_{B^n}^{x^ n}\, \Pi^\delta(\sigma_B|  x^n) \,
\preceq 2^{-n(H(B|  X')_{\sigma}-c'\delta)}
\label{eq:rhonProjIneqCond}
\\
\trace( \Pi^\delta(\sigma_B|  x^n))&\leq 2^{n(H(B|  X')_\sigma+c'\delta)} \label{eq:PidimCond}
\end{align}
where $c'>0$ is a constant,  $\eps'>0$ tends to zero as $\delta\to 0$, $\rho_{B^n}^{x^n}=\bigotimes_{i=1}^n \rho_{B_i}^{x_i}$, and the classical random variable $X'$ is distributed according to the type of $x^n$.
If $x^n\in\tset(p_X)$, then %
\begin{align}
\trace( \Pi^\delta(\sigma_B) \rho_{B^n}^{x^n} )\geq& 1-\eps' 
\label{eq:UnitTCondB}
\end{align}
 as well (see \cite[Property 15.2.7]{Wilde:17b}).
The lemma below is a  simplified version of the quantum packing lemma in %
\cite{HsiehDevetakWinter:08p}.
\begin{lemma}[Quantum Packing Lemma {\cite%
{HsiehDevetakWinter:08p}}]
\label{lemm:Qpacking}
Let %
\begin{align}
\rho=\sum_{x\in\Xset} p_X(x) \rho_x \,,
\end{align}
where $\{ p_X(x), \rho_x \}_{x\in\Xset}$ is a given ensemble. %
Furthermore, suppose that there is  a code projector $\Pi$ and codeword projectors $\Pi_{x^n}$, $x^n\in\tset(p_X)$, that satisfy for every 
$\alpha>0$ and sufficiently large $n$,
\begin{align}
\trace(\Pi\rho_{x^n})\geq&\, 1-\alpha \\
\trace(\Pi_{x^n}\rho_{x^n})\geq&\, 1-\alpha \\
\trace(\Pi_{x^n})\leq&\, 2^{n d}\\
\Pi \rho^{\otimes n} \Pi \preceq&\, 2^{-n(D-\alpha)} \Pi 
\end{align}
for some $0<d<D$ with $\rho_{x^n}\equiv \bigotimes_{i=1}^n \rho_{x_i}$.
Then, there exist codewords $x^n(m)$, $m\in [1:2^{nR}]$, and  a POVM $\{ \Lambda_m \}_{m\in [1:2^{nR}]}$, such that 
\begin{align}
\label{eq:QpackB}
  \trace\left( \Lambda_m \rho_{x^n(m)} \right)  \geq 1-2^{-n[ D-d-R-\eps_n(\alpha)]}
\end{align}
for all %
$m\in [1:2^{nR}]$, where $\eps_n(\alpha)$ tends to zero as $n\rightarrow\infty$ and $\alpha\rightarrow 0$. 
\end{lemma}

\subsection{The Decoupling Theorem}
\label{app:decoupling}
To prove achievability for the quantum capacity theorem, we will use the decoupling theorem \cite{DupuisBertaWullschlegerRenner:14p}.
Before we state the theorem, we give an intuitive explanation in the spirit of \cite[Section 24.10]{Wilde:17b}. %
Consider a quantum channel $\channel_{A\rightarrow B}$ without entanglement assistance.
Let $\ket{\theta_{MK}}$ be a purification of the quantum message state $\theta_M$, where $K$ is Alice's reference system.
Suppose that $\ket{\psi_{K B^n E^n J_1}}$ is a purification of the joint state of Alice's reference system $K$, the channel output  $B^n$, and Bob's environment $E^n$, with a purifying system $J_1$. Observe that if the reduced state $\psi_{K E^n J_1}$ %
is  a product state, i.e. $\psi_{K E^n J_1}=\theta_K \otimes \omega_{E^n J_1}$, then it has a purification of the form
$\ket{\theta_{MK}}\otimes \ket{\omega_{E^n J_1 J_2}}$. Since all purifications are related by isometries \cite[Theorem 5.1.1]{Wilde:17b}, there exists an isometry $D_{B^n\rightarrow M  J_2}$ such that  $\ket{\theta_{MK}}\otimes \ket{\omega_{E^n J_1 J_2}}=   D_{B^n\rightarrow M  J_2} \ket{\psi_{RB^n E^n J_1}}$.
Tracing out $K$, $E^n$, $J_1$, and $J_2$, it follows that there exists a decoding map $\Dset_{B^n\rightarrow M}$ that recovers the message state, i.e.
$\theta_M=\Dset_{B^n\rightarrow M}(\psi_{B^n})$. Therefore, in order to show that there exists a reliable coding scheme, it is sufficient to encode in such a manner that approximately decouples between Alice's reference system and Bob's environment, i.e., such that
$\psi_{KE^n J_1}\approx\theta_K \otimes \omega_{E^n J_1}$.

We will make use of the following definitions from \cite{Renner:08a}. %
 Define the conditional min-entropy by 
\begin{align}
H_{\min}(\rho_{AB}|  \sigma_B)
&=
-\log \inf \left\{ \lambda\in\mathbb{R} \,:\;  \rho_{AB} \preceq \lambda \cdot (\identity_A\otimes \sigma_B)  \right\}
\nonumber\\
H_{\min}(A|  B)_\rho
&=
 \sup_{\sigma_B} H_{\min}(\rho_{AB}|  \sigma_B)
 \,,
\end{align}
where the supremum is over quantum states of the system $B$.
In general, the conditional min-entropy is bounded by
\begin{align}
-\log \abs{\Hset_B} \leq H_{\min}(A|  B)_\rho \leq \log \abs{\Hset_A} \,.
\label{eq:H2bound}
\end{align}
To see this, observe that if we choose $\sigma_B=\frac{\identity_B}{\abs{\Hset_B}}$, then the matrix inequality
$\rho_{AB} \preceq \lambda (\identity_A\otimes \sigma_B)$ holds for $\lambda=\abs{\Hset_B}$,
hence $H_{\min}(\rho_{AB}|  \sigma_B)\geq -\log\abs{\Hset_B}$.
As for the upper bound, %
the matrix inequality implies that
$1=\trace(\rho_{AB})\leq \lambda \abs{\Hset_A}\trace(\sigma_B)=\lambda \abs{\Hset_A}$, hence $H_{\min}(\rho_{AB}|  \sigma_B)\leq 
\log\abs{\Hset_A}$.
Furthermore, %
the lower  bound is saturated when the joint state of $A$ and $B$ is 
$\ket{\Phi_{AB}}$,%
whereas the upper bound for a product state $\frac{\identity_A}{\abs{\Hset_A}}\otimes \rho_B$.

Then, define the smoothed min-entropy by 
\begin{align}
H_{\min}^{\eps}(A|  B)_\rho=\max_{ \sigma_{AB} \,:\; d_F(\rho_{AB},\sigma_{AB})\leq \eps } H_{\min}^{\eps}(A|  B)_\sigma 
\end{align}
for arbitrarily small $\eps>0$, where $d_F(\rho,\sigma)=\sqrt{ 1-\norm{\sqrt{\rho}\sqrt{\sigma}}_1^2 }$ is the fidelity distance between the states. As for the von Neumann entropy, conditioning cannot increase the smoothed min-entropy, \ie
$H_{\min}^{\eps}(A|  BC)_\rho\leq H_{\min}^{\eps}(A|  B)_\rho$ \cite[Lemma 3.1.7]{Renner:08a}.
The theorem below is due to
Dupuis et al.  
\cite{Dupuis:10z,DupuisBertaWullschlegerRenner:14p}).
\begin{theorem}[{Decoupling Theorem \cite{DupuisBertaWullschlegerRenner:14p}}]
\label{theo:decoup1}
Let $\theta_{A_1 K}$ be a quantum state, $\Tset_{A_1\to E}$ a quantum channel, and $\eps>0$  arbitrary. Define
\begin{align}
\omega_{AE}=\Tset_{A_1\to E}(\Phi_{A_1 A}) \,.
\end{align}
Then,
\begin{align}
\int_{\mathbb{U}_{A_1}} \big\lVert \Tset_{A_1\to E}(U_{A_1}\rho_{A_1 K})-\omega_E\otimes \theta_K
\big\rVert_1 \, d{U_{A_1}} \leq
2^{-\frac{1}{2}\left[ H_{\min}^{\eps}(A|  E)_\omega+H_{\min}^{\eps}(A_1|  K)_\theta \right]}+12\eps
\end{align}
where the integral is over the Haar measure on all unitaries $U_A$.
\end{theorem}

The decoupling theorem shows that by choosing a unitary $U_{A}$ uniformly at random, we can approximately  decouple between  $E$ and $K$ provided that $H_{\min}^{\eps}(A|  E)_\omega>-H_{\min}^{\eps}(A_1|  K)_\rho$.
Uhlmann's theorem \cite{Uhlmann:76p} is often used along with the decoupling approach to establish the existence of proper encoding and decoding operations.
\begin{theorem}[Uhlmann's theorem {\cite{Uhlmann:76p}\cite[Corollary 3.2]{Dupuis:10z}}]
\label{theo:Uhlmann}
For every pair of pure states $\ket{\psi_{AB}}$ and $\ket{\theta_{AC}}$ that satisfy $\norm{\psi_A-\theta_A}_1 \leq \eps$, there exists an isometry $F_{B\rightarrow C}$ such that
$ \norm{ (\identity\otimes F_{B\rightarrow C})\psi_{AB}-\theta_{AC} }_1 \leq 2\sqrt{\eps} $.
\end{theorem}

\section{Proof of Lemma~\ref{lemm:pureCea}}
\label{app:pureCea}
Consider the region $\mathcal{R}_{\text{EA}*}(\channel)$ as defined in (\ref{eq:calRClea}).
Fix $\varphi_{A_0 A_1}$, $p_X(x)$, and  $\{ \Fset_{A_0\to A}^{(x)} \}$.  Let  
\begin{align}
R&= I(X;B)_\omega \,,
\label{eq:scR0p}
\\
R'&= I(A_1;B|  X)_\omega \,.
\label{eq:scD0p}
\end{align}
We prove the lemma using similar techniques as in \cite{Pereg:21p, Pereg:19c3}.
It is easy to see that pure states are sufficient, as
every quantum state $\varphi_{A_0 A_1}$ has a purification $\ket{\phi_{A_0 J_0 A_1}}$.
Since $A_0$ is arbitrary, we can extend it and obtain the same characterization when $A_0$ is replaced by $\bar{A}_0=(A_0,J_0)$. 

To bound the alphabet size of the random variable $X$, we use the Fenchel-Eggleston-Carath\'eodory lemma \cite{Eggleston:66p} and similar arguments as in previous works \cite{YardHaydenDevetak:08p,Pereg:21p}.
Having fixed $\varphi_{A_0 A_1}$ and  $\{ \Fset_{A_0\to A}^{(x)} \}$,
define 
\begin{align}
\omega_{A}^x\equiv \Fset^{(x)}_{A_0\to A}(\varphi_{A_0 }) \,.
\end{align}
Every quantum state $\rho_A$ has a unique and real parametric representation $u(\rho_A)$ of dimension $\abs{\Hset_A}^2-1$ (see \cite[Appendix B-B]{Pereg:21p}). 
Then, define a map $g:\Xset\rightarrow \mathbb{R}^{\abs{\Hset_A}^2+1}$ by
\begin{align}%
g(x)= \big(  u(\omega_A^x) \,,\; H(B|  X=x)_{\omega} \,,\; 
I(A_1;B|  X=x)_\omega  \big) \,.
\end{align}%
 
The map $g$ can be extended to probability distributions as follows:
\begin{align}%
G \,:\; p_X  \mapsto
\sum_{x\in\Xset} p_X(x) g(x)=
 \big(  u(\omega_A) \,,\; H(B|  X)_{\omega} \,,\; I(A_1;B|  X)_\omega    \big)  \,,
\end{align}%
as $\omega_A=\sum_x p_X(x) \omega_A^x$.
According to the Fenchel-Eggleston-Carath\'eodory lemma \cite{Eggleston:66p}, any point in the convex closure of a connected compact set within $\mathbb{R}^d$ belongs to the convex hull of $d$ points in the set. 
Since the map $G$ is linear, it maps  the set of distributions on $\Xset$ to a connected compact set in $\mathbb{R}^{\abs{\Hset_A}^2+1}$. Thus, for every  $p_X$, 
there exists a probability distribution $p_{\bar{X}}$ on a subset $\overline{\Xset}\subseteq \Xset$ of size $%
\abs{\Hset_A}^2+1$, such that 
$G(p_{\bar{X}})=G(p_{X}) $. 
We deduce that the cardinality of $X$ can be restricted to $\abs{\Xset}\leq \abs{\Hset_A}^2+1$, while preserving $\omega_A$,  and thus,
the output state $\omega_B\equiv \channel(\omega_A)$ as well, and the mutual informations $I(X;B)_\omega=H(B)_\rho-H(B|  X)_{\omega}$ and $I(A_1;B|  X)_\omega$.

This completes the proof of the lemma.
\qed

\section{Proof of Theorem~\ref{theo:ClEA}}
\label{app:ClEA}
Consider classical communication over a quantum channel $\channel_{A\to B}$ with unreliable entanglement assistance.

\subsection{Achievability Proof}
We show that for every $\eps_0,\delta_0>0$, there exists a $(2^{n(R-\delta_0)},2^{n(R'-\delta_0)},n,\eps_0)$ code with unreliable entanglement assistance, provided that $(R,R')\in\inR_{\text{EA*}}(\channel)$. 
To this end, we will use the quantum packing lemma \cite{HsiehDevetakWinter:08p}, as presented in Subsection~\ref{app:packing} of the previous appendix. 
Recall that by Lemma~\ref{lemm:pureCea}, it suffices to consider pure states.
Then, let $\ket{\phi_{G_1 G_2} }$ be a  pure entangled state on $\Hset_{A_0}\otimes \Hset_{A_0}$, and $\Fset^{(x)}$ be a quantum channel acting on $\mathscr{S}(\Hset_{A_0})$ (see (\ref{eq:calRClea})-(\ref{eq:calRClea2})).  Suppose that Alice and Bob share $\ket{\phi_{G_1 G_2} }^{\otimes n}$.

As we explain below, we can restrict ourselves to isometric encoding maps.
For a moment, let us denote the channel input by $S$, and consider the channel $\channel_{S\to B}$.
In the derivation below, we will use the encoding channel $\Fset^{(x)}_{A_0\to S}$.
Every quantum channel $\Fset^{(x)}_{A_0\to S}$ has an isometric extension $F^{(x)}_{A_0\to S \breve{S}}$. Since it is an encoding mapping, we may as well take $\breve{S}$ to be Alice's ancilla. Then, let $A\equiv S\breve{S}$ be the augmented channel input. We are effectively coding over the  channel $\tilde{\channel}_{A\to B}$, which is defined by
\begin{align}
\tilde{\channel}_{A\to B}(\rho_{S\breve{S}})\equiv \channel_{S\to B}(\trace_{\breve{S}}(\rho_{S\breve{S}})) \,,
\label{eq:tchannel}
\end{align}
using the isometric map $F^{(x)}_{A_0\to A}$.

From this point, we will focus on the quantum channel $\tilde{\channel}_{A\to B}$ and use the encoding isometry $F^{(x)}_{A_0\to A}$.
Define
\begin{align}
\ket{\psi^x_{A G_2 }}&= (F^{(x)}\otimes\identity) \ket{\phi_{G_1 G_2}} %
\,,\label{eq:ABphis} \\
\omega_{B G_2}^x&= ( \tilde{\channel}\otimes \text{id})(\psi^x_{A G_2}) \,.
\label{eq:oBpB}
\end{align}
We will often use the notation $\ket{\psi^{x^n}}\equiv \bigotimes_{i=1}^n \ket{\psi^{x_i}}$.
The code construction, encoding and decoding procedures are described below.

\vspace{0.1cm}
\subsubsection{Code Construction}
First, consider a classical codebook.
Let $\{x^n(m)\}_{m\in [1:2^{nR}]}$ be a set of $2^{nR}$ classical codewords  that %
will be chosen later, in Subsection~\ref{subsec:errNoEA} below.
We define the encoding operators in terms of this classical codebook.
Denote the Heisenberg-Weyl operators of dimension $D$ by $\{ \Sigma(a,b)= \Sigma_X^a\Sigma_Z^b \}$, where
$%
\Sigma_X= \sum_{j=0}^{D-1} \ketbra{ j\oplus 1}{ j } %
$ and $\Sigma_Z= \sum_{j=0}^{D-1} e^{2\pi ij/D} \kb{j}
$, %
for $a,b\in \{ 0,1,\ldots, D-1  \}$, with
$ j\oplus k= (j+k) \mod D \;$ and
$i=\sqrt{-1}$.

Consider a Schmidt decomposition of the pure state in (\ref{eq:ABphis}),
\begin{align}
\ket{\psi_{A G_2}^x}=\sum_{z\in\Zset} \sqrt{p_{Z|  X}(z|  x)} \ket{\xi_{z,x}} \otimes \ket{\xi_{z,x}'} \,,
\end{align}
where $p_{Z|  X}$ is a conditional probability distribution, while $\{\ket{\xi_{z,x}} \}$  and $\ket{\xi'_{z,x}}$ are orthonormal sets.
For every $x^n\in\Xset^n$ and every conditional type class $\Tset_n(t|  x^n)$ in $\Zset^n$, define the operators 
\begin{align}
&V_{t}(a_t,b_t,c_t)= (-1)^{c_t} \Sigma(a_t,b_t) \,,\; \nonumber\\&
 a_t,b_t\in \{ 0,1,\ldots, D_t-1  \} \,,\; c_t\in\{0,1\} \,,
\label{eq:Vt}
\end{align}
where $D_t=\abs{\Tset_n(t|  x^n)}$ is the size of type class associated with the conditional type $t$. Then, define the operator
\begin{align}
U(\gamma)=\bigoplus_t  \, V_{t}(a_t,b_t,c_t) 
\label{eq:Ugamma}
\end{align}
with $\gamma=\left( (a_t,b_t,c_t)_t \right)$. Let $\Gamma_{x^n}$ denote the set of all possible vectors $\gamma$. 

For every $m\in [1:2^{nR}]$, choose 
$2^{nR'}$ vectors $\gamma(m'|  x^n(m))$, $m'\in [1:2^{nR'}]$,  uniformly at random from $\Gamma_{x^n(m)}$.
The mappings are revealed to both Alice and Bob.

\vspace{0.1cm}
\subsubsection{Encoder}
To send the messages $(m,m')\in [1:2^{nR}]\times [1:2^{nR'}]$, apply the operators $F^{(x_i(m))}$ and   
$U(\gamma(m'|  x^n(m)))$. This yields the input state
\begin{align}
\ket{\chi_{A^n G_2^n}^{\gamma,x^n}}\equiv (  U(\gamma) F^{(x^n)} \otimes\identity) \ket{\phi_{G_1 G_2} }^{\otimes n} \,,
\end{align}
for $x^n=x^n(m)$ and $\gamma=\gamma(m'|  x^n)$,
where $F^{(x^n)}\equiv \bigotimes_{i=1}^n F^{(x_i)}$. %
Then, transmit $A^n$ through the channel.

\vspace{0.1cm}
\subsubsection{Decoder}
Bob receives the systems $B^n$ in a state 
$\rho^{\gamma,x^n}_{B^n G_2^n}%
$,  
and decodes as follows.
\begin{enumerate}[(i)]
\item
Measure $B^n$ using a POVM $\{ \Lambda_m \}_{m\in [1:2^{nR}]}$. Denote the measurement outcome by $\hm$.
\item
If there is no entanglement assistance, declare $\hm$ as the message estimate.
\item
If entanglement assistance is present, measure $B^n G_2^n$ jointly using a second POVM $\{ 
\Upsilon_{m'|  x^n(\hm)} \}_{m'\in [1:2^{nR'}]}$. Let $\hm'$ be the outcome of this measurement.
Then, declare  $(\hm,\hm')$ as the estimated message pair.
\end{enumerate}
The POVMs $\{ \Lambda_m \}$ and $\{\Upsilon_{m'|  x^n(\hm)}\} $  will be chosen later in Subsections \ref{subsec:errNoEA} and 
\ref{subsec:errEA}, respectively.

\vspace{0.1cm}
\subsubsection{Code Properties}
Before we go into the error analysis, we show that Alice's operations for encoding the second message $m'$ can be effectively reflected to Bob's side. %
To this end, we will apply %
the ``ricochet property" \cite[Eq. (17)]{HsiehDevetakWinter:08p},
\begin{align}
(U\otimes\identity)\ket{\Phi_{AB}}=(\identity\otimes U^T)\ket{\Phi_{AB}} \,.
\label{eq:ricochet}
\end{align}

Now, for every $x^n\in\Xset^n$, 
\begin{align}
\ket{\psi^{x^n}_{A^n G_2^n} }
&= 
\sum_{z^n\in\Zset^n} \sqrt{ p_{Z^n|  X^n}(z^n|  x^n) }  \ket{\xi_{x^n,z^n}} \otimes \ket{\xi'_{x^n,z^n} } \,,
\end{align}
where $p_{Z^n|  X^n}(z^n|  x^n)=\prod_{i=1}^n p_{Z|  X}(z_i|  x_i)$. As the space $\Zset^n$ can be partitioned into conditional type classes given $x^n$, we may write 
\begin{align}
\ket{\psi^{x^n}_{A^n G_2^n} }
&= \sum_{t\in\Pset_n(\Zset)} \sum_{z^n\in \Tset_n(t|  x^n)}
 \sqrt{ p_{Z^n|  X^n}(z^n|  x^n) } 
\ket{\xi_{x^n,z^n} } \otimes \ket{\xi'_{x^n,z^n} }
\nonumber\\
&=  \sum_{t\in\Pset_n(\Zset)} \sqrt{ p_{Z^n|  X^n}(z_t^n|  x^n) } \sum_{z^n\in \Tset_n(t|  x^n)}  
 \ket{\xi_{x^n,z^n}} \otimes \ket{\xi'_{x^n,z^n} } \,,
\end{align}
where $z_t^n$ is any sequence in the conditional type class $\Tset_n(t|  x^n)$. Therefore, 
\begin{align}
\ket{\psi^{x^n}_{A^n G_2^n} }
=&  \sum_{t\in\Pset_n(\Zset)} \sqrt{ P(t|  x^n)}   \ket{\Phi_t} \,,
\label{eq:phiKBsnT}
\end{align}
where 
$P(t|x^n)=P_{Z^n|X^n}(z_t^n|x^n) |\Tset_n(t|x^n)|$ is the conditional probability  of the type class $\Tset_n(t|s^n)$, and 
\begin{align}
\ket{\Phi_t }=
\frac{1}{\sqrt{|\Tset_n(t|x^n)|%
}} \sum_{z^n\in \Tset_n(t|  x^n)}   \ket{\xi_{x^n,z^n}} \otimes \ket{\xi'_{x^n,z^n}} \,.
\end{align}

Alice applies the operator $U(\gamma(m'|  x^n(m)))$ to the entangled states.
Since the state $\ket{\Phi_t }$ is  maximally entangled, we have by %
(\ref{eq:ricochet}), %
\begin{align}
\ket{\chi_{A^n G_2^n}^{\gamma,x^n}}&\equiv (U(\gamma(m,m'))\otimes \identity)
\ket{\psi^{x^n}_{A G_2^n} } \nonumber\\&=
(\identity\otimes U^T(\gamma(m,m')))\ket{\psi^{x^n}_{A^n G_2^n} } \,,
\label{eq:RicoCons}
\end{align}
where $\ket{\psi^{x^n}_{A^n G_2^n} }\equiv \bigotimes_{i=1}^n \ket{\psi^{x_i}_{A G_2} }$ (see (\ref{eq:ABphis})).
By the same considerations, %
\begin{align}
\ket{\psi^{x^n}_{A^n G_2^n} } &= (F^{(x^n)}\otimes \identity) \ket{\phi_{G_1 G_2}}^{\otimes n}  \nonumber\\&=
(\identity \otimes (F^{(x^n)})^T ) \ket{\phi_{A G_2'}}^{\otimes n} \,.
\end{align}
That is, Alice's unitary operations can be reflected and treated as if performed by Bob.

Bob then receives the systems $B^n$ in the  state
\begin{align}
\rho^{\gamma,x^n}_{ B^n G_2^n}&=
   (\tilde{\channel}^{\otimes n}\otimes\text{id})  \left( 
 \chi_{A^n G_2^n}^{\gamma,x^n} 
\right)
\label{eq:rhoG1}
\\
&=
 (\tilde{\channel}^{\otimes n}\otimes\text{id}) \left( 
(\identity\otimes U^T(\gamma)) %
\psi_{A^n G_2^n}^{x^n} %
(\identity\otimes U^*(\gamma))
\right)\,,
\label{eq:rhoG1c}
\end{align}
where $\tilde{\channel}$ is as in (\ref{eq:tchannel}), and the last line is due to  (\ref{eq:RicoCons}). 
Since a quantum channel is a linear map,  the above can be written as
\begin{align}
\rho^{\gamma,x^n}_{B^n G_2^n}
=& 
(\identity\otimes U^T(\gamma)) %
(\tilde{\channel}\otimes\text{id})(  \psi_{A^n G_2^n}^{x^n}  ) %
(\identity\otimes U^*(\gamma))
\nonumber\\
=& 
(\identity\otimes U^T(\gamma))  \omega_{B^n G_2^n}^{x^ n}
(\identity\otimes U^*(\gamma)) \,.
\label{eq:rhoG2}
\end{align}

\vspace{0.1cm}
\subsubsection{Error Analysis Without Assistance}
\label{subsec:errNoEA}
Recall that if entanglement assistance is absent, then Bob does not decode $m'$. 
Furthermore, since the decoder cannot measure $G_2^n$ in this case, we need to consider the reduced state $\rho_{B^n}^{\gamma,x^n}$ of the joint output state $\rho_{B^n G_2^n}^{\gamma,x^n}$. 
Observe that by (\ref{eq:rhoG2}), the reduced output state is 
\begin{align}
\rho_{B^n}^{\gamma,x^n(m)}=\omega_{B^n}^{x^n(m)} \,.
\end{align}
Thereby, the reduced output is not affected by the encoding of the message $m'$ using $\gamma(m'|  x^n(m))$, and we can use the standard results on classical communication over a quantum channel without assistance. 

Fix $\delta>0$.
Based on the HSW Theorem  \cite{Holevo:98p,SchumacherWestmoreland:97p}, there exists a codebook  
$\{ x^n(m) \}$  and a
 POVM $\{\Lambda_m\}$ such that
\begin{align}
P_{e|  m,m'}^{*(n)}(\Fset,\phi_{G_1 G_2}^{\otimes n},\Lambda) 
&=1-\trace(\Lambda_m \rho_{B^n}^{\gamma(m'|  x^n(m)),x^n(m)})
\nonumber\\
&=1-\trace(\Lambda_m \omega_{B^n}^{x^n(m)})
\nonumber\\
&\leq 2^{-n(I(X;B)_\omega-R-\eps_1(\delta))}  \,,
\label{eq:B2a}
\end{align}
where $x^n(m)\in\tset(p_X)$ for all $m\in [1:2^{nR}]$ (see  \cite[Section 20.3.1]{Wilde:17b}). We use the notation $\eps_j(\delta)$ for terms that tend to zero as $\delta\to 0$.
Thus, in the absence of entanglement assistance, the probability of error tends to zero  as $n\to \infty$,   provided that
\begin{align}
R<I(X;B)_\omega-\eps_1(\delta) \,.
\label{eq:RB2a}
\end{align}

\vspace{0.1cm}
\subsubsection{Packing Lemma Requirements}
\label{subsec:errEA}
In the error analysis with entanglement assistance, we will use the quantum packing lemma. Fix a sequence $x^n\in\tset(p_X)$.
 Consider the ensemble 
$\{ p(\gamma)=\frac{1}{\abs{\Gamma_{x^n}}},$ $  \rho^{\gamma,x^n}_{B^n, G_2^n} \}$, for which the expected density operator is 
\begin{align}
\overline{\rho}^{x^n}_{B^n G_2^n}=\frac{1}{\abs{\Gamma_{x^n}}}\sum_{\gamma\in\Gamma_{x^n}} \rho^{\gamma,x^n}_{B^n G_2^n} \,.
\label{eq:SigmaBpB}
\end{align}
Define the code projector and the codeword projectors by
\begin{align}
\Pi 
&
\equiv \Pi^{\delta}(\omega_{B}|  x^n)\otimes \Pi^{\delta}(\omega_{G_2}|  x^n)
\label{eq:CodeProjDir}
\\
\Pi_\gamma 
&
\equiv (\identity\otimes U^T(\gamma)) \Pi^{\delta}(\omega_{B G_2}|  x^n) (\identity\otimes U^*(\gamma)) \,,\; 
\text{for 
$\gamma\in\Gamma_{x^n}$} \,,
\end{align}
where $\Pi^{\delta}(\omega_{BG_2}|  x^n)$, $\Pi^{\delta}(\omega_{B}|  x^n)$ and $\Pi^{\delta}(\omega_{G_2}|  x^n)$ are the projectors onto the conditional $\delta$-typical subspaces associated with the states $\omega^x_{BG_2}$, $\omega^x_{B}=\trace_{G_2}(\omega^x_{BG_2})$ and
$\omega^x_{G_2}=\trace_{B}(\omega^x_{BG_2})$, respectively (see (\ref{eq:oBpB})).
Applying the bounds in  \cite[Appendix II]{HsiehDevetakWinter:08p} to the operators above, we obtain
\begin{align}
\trace(\Pi \rho^{\gamma,x^n}_{B^n G_2^n})
&\geq 1-2\eps_2(\delta) \,,
\\
\trace(\Pi_\gamma \rho^{\gamma,x^n}_{B^n G_2^n})& \geq 1-\eps_3(\delta) \,,
\\
\trace(\Pi_\gamma) &\leq 2^{n(H(BG_2|  X)_\omega+\eps_4(\delta))} \,,
\\
\Pi \sigma_{B'^n G_2^n}^{x^n}  \Pi &\preceq 
2^{-n(H(B|  X)_\omega+H(G_2|  X)_\omega+\eps_5(\delta))} \Pi \,,
\end{align}
where $\eps_j(\delta)$ tend to zero as $\delta\to 0$.
Hence,  the requirements of the packing lemma are satisfied. Then, by Lemma~\ref{lemm:Qpacking}, there exist deterministic vectors $\gamma(m'|  x^n)$, $m'\in [1:2^{nR'}]$, and  a POVM 
$\{ \Upsilon_{m'|  x^n} \}_{m'\in [1:2^{nR'}]}$, such that 
\begin{align}
	\trace\left( \Upsilon_{m'|  x^n} \rho^{\gamma(m'|  x^n),x^n}_{ B^n G_2^n} \right)  
	&\geq 1-2^{-n[ I(B;G_2|  X)_\omega-R'-\eps_6(\delta)]}
	\label{eq:packDirIneq}
\end{align}
for all %
$m'\in [1:2^{nR'}]$. 

\vspace{0.1cm}
\subsubsection{Error Analysis with Entanglement Assistance}
Suppose that entanglement assistance is present, in which case Bob estimates both $m$ and $m'$.
Hence, the error event is bounded by the union of the following events,
\begin{align}
\mathscr{E}_1(m)=& \{  \hm\neq m %
\}\,, \\
\mathscr{E}_2(m')=& \{  \hm'\neq m' %
\} \,.
\end{align}
Then, by the union of events bound,
\begin{align}%
P_{e|  m,m'}^{(n)}(\Fset,\phi_{G_1 G_2}^{\otimes n},\Upsilon\circ\Lambda)\leq 
\prob{  \mathscr{E}_1(m)  }+
\cprob{  \mathscr{E}_2(m')  }{ \mathscr{E}_1^c(m)}
 \,.
\label{eq:Edir1}
\end{align}%
The first term corresponds to the measurement of $\{\Lambda_m\}$ above.
Based on our previous analysis (see (\ref{eq:B2a})),
\begin{align}
\prob{\mathscr{E}_1(m)}
&\leq 2^{-n(I(X;B)_\omega-R-\eps_1(\delta))}
\,,
\label{eq:B2a2}
\end{align}
which tends to zero for a rate $R$ as in (\ref{eq:RB2a}).

Given $ \mathscr{E}_1^c$, %
Bob has recovered the correct $m$ in step (i) of the decoding procedure. Denote the joint state of the systems $B^n G_2^n$ after this measurement by $\widetilde{\rho}^{\,\gamma,x^n(m)}_{B^n G_2^n}$.
As previously observed \cite{Pereg:21p,Pereg:20c1},  by %
the gentle measurement lemma \cite{Winter:99p,OgawaNagaoka:07p} and (\ref{eq:B2a2}), %
the post-measurement state is close to the original state in the sense that
\begin{align}
\frac{1}{2}\norm{\widetilde{\rho}^{\,\gamma,x^n(m)}_{B^n G_2^n}-\rho^{\gamma,x^n(m)}_{B^n G_2^n}}_1 
&
\leq 2^{ -n\frac{1}{2}( I(X;B)_\omega-R-\eps_1(\delta)) }
\nonumber\\&
 \leq \eps_7(\delta)
\end{align}
for sufficiently large $n$ and $R$ as in (\ref{eq:RB2a}).
Therefore, the distribution of measurement outcomes, when $\widetilde{\rho}^{\,\gamma,x^n}_{B^n G_2^n}$ is measured, is roughly the same as if the POVM 
$\Lambda_{m}$ was never performed. To be precise, the difference between the probability of a measurement outcome $\hm'$ when $\widetilde{\rho}^{\,\gamma,x^n}_{B^n G_2^n}$ is measured and the probability when $\rho^{\,\gamma,x^n}_{B^n G_2^n}$ is measured is bounded by $ \eps_7(\delta)$ in absolute value (see  \cite[Lemma 9.11]{Wilde:17b}).
 Therefore, by (\ref{eq:packDirIneq}), the POVM $\Upsilon_{m'|  x^n(m)}$ satisfies
$%
\cprob{ \mathscr{E}_2 }{ \mathscr{E}_1^c} \leq 2^{ -n( I(G_2;B|  X)_\rho -R'-\eps_8(\delta)) } 
$, %
which tends to zero as $n\rightarrow\infty$, if
\begin{align}
R'< I(G_2;B|  X)_\omega -\eps_8(\delta) \,.
\label{eq:B3}
\end{align}
Finally, we let  $A_0$, $A_1$ replace $G_1$, $G_2$, respectively. %
Thus, the probability of error tends to zero as $n\rightarrow\infty$ provided that  $R<I(X;B)_\omega-\eps_1(\delta)$ and 
$R'< I(A_1;B|  X)_\omega -\eps_8(\delta)$.
This completes the proof of the direct part.

\subsection{Converse Proof}
Suppose that Alice and Bob are trying to distribute randomness. An upper bound on the rate at which Alice can distribute randomness to Bob also serves as an upper bound on the rate at which they can communicate. In this task, Alice and Bob share an unreliable entangled resource $\Psi_{G_A G_B}$. Alice first prepares  maximally corrleated states,
\begin{align}
\pi_{K M K'M'} &\equiv 
\left(\frac{1}{2^{nR}}\sum_{m=1}^{2^{nR}} \kb{ m } \otimes \kb{ m }\right) 
\otimes
\left(\frac{1}{2^{nR'}}\sum_{m'=1}^{2^{nR'}} \kb{ m' } \otimes \kb{ m' }\right) %
\end{align}
locally. That is, $K$, $M$, $K'$, and $M'$ are classical registers that store uniformly-distributed indices $m$ and $m'$.
 Then, Alice applies an encoding channel  to the classical system $M M'$ and her share $G_A$ of the entangled state $\Psi_{G_A G_B}$. As Alice applies $\Fset_{M M'G_A \rightarrow A^n}$, 
the resulting state is 
\begin{align}
\sigma_{ K K' A^n G_B }\equiv (\text{id}\otimes\Fset\otimes\text{id})( \pi_{K K' M M'}\otimes \Psi_{G_A G_B} ) \,.
\end{align}
After Alice sends the system $A^n$ through the channel, Bob receives the system $B^n$ in the state
\begin{align}
\omega_{K K'  G_B B^n}\equiv (\text{id}\otimes\channel^{\otimes n}) (\sigma_{ K K'   G_B A^n}) \,.
\end{align}
In the presence of entanglement assistance, 
 Bob performs a decoding channel $\Dset_{B^n G_B\rightarrow \hM \hM'}$, which yields
\begin{align}
\rho_{ K K' \hM \hM'}\equiv (\text{id}\otimes\Dset)(\omega_{K K' B^n G_B}) \,.
\label{eq:DecConv1EA}
\end{align}
If the entanglement assistance is not available, 
 Bob performs  $\Dset^*_{B^n \rightarrow \tilde{M} }$, producing 
\begin{align}
\rho^*_{ K K' \tilde{M}  G_B}\equiv (\text{id}\otimes\Dset^*\otimes\text{id})(\omega_{K K' B^n G_B}) \,.
\label{eq:DecConv1noEA}
\end{align}

Consider a sequence of codes $(\Eset_n,\Psi_n,\Dset_n,\Dset_n^*)$ of randomness distribution with unreliable entanglement assistance, such that %
\begin{align}
&\frac{1}{2} \norm{ \rho_{K \hM K'\hM'} -\pi_{K M K'M'} }_1 \leq \alpha_n \,,\\
&\frac{1}{2} \norm{ \rho^*_{K \tilde{M} } -\pi_{K M } }_1 \leq \alpha_n^* \,,
\end{align}
where $\rho^*_{K \tilde{M}}$  is the reduced density operator of  $\rho^*_{ K K' \tilde{M}  G_B}$, 
 while $\alpha_n,\alpha_n^*$ tend to zero as $n\rightarrow\infty$.
By the Alicki-Fannes-Winter inequality \cite{AlickiFannes:04p,Winter:16p} \cite[Theorem 11.10.3]{Wilde:17b}, this implies 
\begin{align}
\abs{H(K'|  \hM' K)_\rho - H(K'|  M' K)_{\pi} }\leq n\eps_n \,, \label{eq:AFW1}\\
\abs{H(K|  \tilde{M})_{\rho^*} - H(K|  M)_{\pi} }\leq n\eps_n^* \,,
\label{eq:AFW2}
\end{align}
where $\eps_n,\eps_n^*$ tends to zero as $n\rightarrow\infty$.

Now, suppose that entanglement assistance is absent. Observe that %
$H(K)_{\rho^*}=H( K)_{\pi}=nR$ implies  
$I(K;M)_{\pi} - I(K;\hM)_{\rho^*}= H(K|  \hM)_{\rho^*} - H(K|  M)_{\pi}$. Therefore, by (\ref{eq:AFW2}),
\begin{align}
nR=&I(K;M)_{\pi} \nonumber\\
\leq& I(K;\hM)_{\rho^*}+n\eps_n^* \nonumber\\
\leq& I(K;B^n)_{\omega}+n\eps_n^* \,,
\label{eq:ConvIneq1noEA}
\end{align}
where the last line follows from (\ref{eq:DecConv1noEA}) and the quantum data processing inequality \cite[Theorem 11.5]{NielsenChuang:02b}. Here, the system $G_B$ need not be included since the decoding measurement $\Dset^*$ is only applied to $B^n$.

We move to the case where entanglement assistance is present. Similarly, %
$I(K';M'|  K)_{\pi} - I(K';\hM'|  K)_{\rho}= H(K'|  \hM'K)_{\rho} - H(K'|  M'K)_{\pi}$. Therefore, by (\ref{eq:AFW1}),
\begin{align}
nR'=&I(K';M'|  K)_{\pi} \nonumber\\
\leq& I(K';\hM'|  K)_{\rho}+n\eps_n \nonumber\\
\leq& I(K';G_B B^n|  K)_{\omega}+n\eps_n 
\label{eq:ConvIneq1noEAp}
\end{align}
by (\ref{eq:DecConv1EA}) and the quantum data processing inequality. %
We must include the entanglemet resource system $G_B$, since the decoder measures $G_B B^n$.
By the chain rule, the last bound can also be written as
\begin{align}
nR' &\leq I(K' G_B;B^n|  K)_{\rho}-I(G_B;B^n|  K)_{\rho} 
+I(K';G_B|  K)_{\rho}+n\eps_n \nonumber\\
&\leq I(K' G_B;B^n|  K)_{\rho}+I(K';G_B|  K)_{\omega}+n\eps_n \nonumber\\
&= I(K' G_B;B^n|  K)_{\omega}+n\eps_n \,,
\label{eq:ConvIneq2}
\end{align}
where the equality holds since $G_B$ and $(K,K')$  are in a product state. 
To complete the regularized converse proof, set $X^n=f(K)$ and $A_1^n\equiv (K',G_B)$, where $f$ is an arbitrary one-to-one function from $[1:2^{nR}]$ to $\Xset^n$.
This concludes the proof of Theorem~\ref{theo:ClEA}.
\qed %

\section{Proof of Theorem~\ref{theo:qEA}}
\label{app:qEA}
Consider quantum communication over $\channel_{A\to B}$ with unreliable entanglment assistance.

\subsection{Achievability Proof}
In the proof we follow Dupuis' methods \cite{Dupuis:10z}, originally applied to the quantum broadcast channel.

At first we restrict the entanglement resources to a given rate $R_e$. That is, 
we assume %
$\abs{\Hset_{G_A}}=\abs{\Hset_{G_B}}\leq 2^{nR_e}$.
We are going to show that any rate pair $(Q,Q')$ is achievable with unreliable entanglement assistance if 
\begin{subequations}
\begin{align}
&0\leq Q< H(A_1|  A_2)_\omega \\
& Q<I(A_1\rangle B)_\omega \\
& Q+Q'+R_e<H(A_2)_\omega \\
& Q+Q'-R_e<I(A_2\rangle B)_\omega
\end{align}
\label{eq:QrateAch}
\end{subequations}
for some $\varphi_{A_1 A_2 A}$, where $A_1$, $A_2$ are arbitrary systems, and $\omega_{A_1 A_2 B}=\channel_{A\to B}$ $(\varphi_{A_1 A_2 A})$.

Let $\ket{\phi_{A_1 A_2 A  J}}$ be a purification of $\varphi_{A_1 A_2 A}$. %
Then, the corresponding channel output is 
\begin{align}
\ket{\omega_{A_1 A_2 B E J}} = U^{\channel}_{ A\to B E} \ket{\phi_{A_1 A_2  A  J}} \,,
\label{eq:Oacbkj}
\end{align}
where $\Uset^{\channel}_{ A\rightarrow B E}$ is a Stinespring dilation such that
$%
\Uset^{\channel}_{ A\rightarrow B E}(\rho_{A})=$ $U^{\channel}_{ A\rightarrow B E} \rho_{ A}$ $ (U^{\channel }_{ A\rightarrow B E})^\dagger
$. %

Given a quantum message state $\theta_M\otimes \xi_{\bar{M}}$, let
$K$ and $\bar{K}$ be reference systems that purify the message systems $M$ and $\bar{M}$, respectively, \ie such that the systems $M$, $\bar{M}$, $K$, and $\bar{K}$ have a pure joint state
$\ket{\theta_{MK}}\otimes \ket{\xi_{\bar{M}\bar{K}}}$, with $\abs{\Hset_K}=\abs{\Hset_M}=2^{nQ}$ and $\abs{\Hset_{\bar{K}}}=\abs{\Hset_{\bar{M}}}=2^{n(Q+Q')}$. 
Suppose that given reliable entanglement assistance, Alice and Bob share an entangled state $\ket{\Phi_{G_A G_B}}$ of dimension $\abs{\Hset_{G_A}}=\abs{\Hset_{G_B}}=2^{nR_e}$.

Let $V^{(1)}_{M\to A_1^n}$ and $V^{(2)}_{\bar{M} G_A\to A_2^n}$ be arbitrary full-rank partial isometries. That is, each operator has $0$-$1$ singular values with a rank of $2^{nQ}$ and $2^{n(Q+Q')}$, respectively. 
Denote
\begin{align}
\ket{\psi^{(1)}_{A_1^n K}}&= V^{(1)}_{M\to A_1^n} \ket{\theta_{M K}} \,, \\
\ket{\psi^{(2)}_{A_2^n G_B \bar{K}}}&= V^{(2)}_{\bar{M} G_A\to A_2^n} (\ket{\xi_{\bar{K}\bar{M}}} \otimes \ket{\Phi_{G_A,G_B}})
\,.
\end{align}

\subsubsection{Decoupling Inequalities}
First, we establish decoupling inequalities. %
We introduce the following notation of operators and channels.
 For every pair of Hilbert spaces $\Hset_{A}$ and $\Hset_B$ with orthonormal bases $\{ \ket{i_A} \}$ and $\{ \ket{j_B} \}$, respectively, define the operator $\text{op}_{A \rightarrow B}(\ket{\psi_{AB}})$ by
\begin{align}
\text{op}_{A \rightarrow B}(\ket{i_A} \otimes \ket{j_B} )\equiv \ketbra{j_B }{ i_A} \,.
\label{eq:op}
\end{align}
While the operation above depends on the choice of bases, we will not specify these since it is not important for our purposes.
To generalize this definition to any state $\ket{\psi_{AB}}$, consider its decomposition 
$\ket{\psi_{AB}}=\sum_{i,j} a_{i,j} \ket{i_A} \otimes \ket{j_B} $, and define 
 $\text{op}_{A \rightarrow B}(\ket{\psi_{AB}})=\sum_{i,j} a_{i,j} \text{op}_{A \rightarrow B}(\ket{i_A} \otimes \ket{j_B} )$.

Consider the operators
\begin{align}
\Pi_{A_2\to A_1 AJ}=
\sqrt{\abs{\Hset_{A_2}}}  \text{op}_{A_2\to A_1 AJ}(\phi_{A_1 A_2 AJ}) \,, \\
\Pi_{A_1\to A_2 AJ}=
\sqrt{\abs{\Hset_{A_1}}}  \text{op}_{A_1\to A_2 AJ}(\phi_{A_1 A_2 AJ}) \,.
\end{align}
Given a pair of unitaries, $U^{(1)}_{A_1^n}$ and $U^{(2)}_{A_2^n}$, define the following quantum states,
\begin{align}
\ket{\omega^{U^{(2)}}_{A_1^n A^n  J^n \bar{K} G_B}}
&
= \Pi_{A_2 \to A_1 A J}^{\otimes n} U^{(2)}_{A_2^n}
\ket{\psi^{(2)}_{A_2^n G_B \bar{K}}}%
\,,
\\
\ket{\omega^{U^{(1)}}_{A_2^n  A^n J^n K}}
&
= \Pi_{A_1 \to A_2 A J}^{\otimes n} U^{(1)}_{A_1^n} 
\ket{\psi^{(1)}_{A_1^n K}}%
\,.
\end{align}
The corresponding channel outputs are then 
\begin{align}
\ket{\omega^{U^{(2)}}_{A_1^n B^n E^n J^n \bar{K} G_B}}&= (U_{A\to BE}^{\channel})^{\otimes n} \ket{\omega^{U^{(2)}}_{A_1^n A^n  J^n \bar{K} G_B}} ,
\\
\ket{\omega^{U^{(1)}}_{A_2^n  B^n E^n J^n K}}&= (U_{A\to BE}^{\channel})^{\otimes n}\ket{\omega^{U^{(1)}}_{A_2^n  A^n J^n K}} \,.
\end{align}

Now, consider the operators
\begin{align}
\Pi^{U^{(2)}}_{A_1^n\to A^n J^n \bar{K} G_B}
&
= 
\sqrt{\abs{\Hset_{A_1}}^n}  \text{op}_{A_1^n\to A J \bar{K} G_B}(\omega^{U^{(2)}}_{A_1^n A^n  J^n \bar{K} G_B}) \,, \\
\Pi^{U^{(1)}}_{A_2^n \to A^n J^n K}
&
= 
\sqrt{\abs{\Hset_{A_2}}^n}  \text{op}_{A_2^n \to A^n J^n K}(\omega^{U^{(1)}}_{A_2^n  A^n J^n K}) \,, \\
\Pi_{A_1 A_2 \to A J }
&
= 
\sqrt{\abs{\Hset_{A_1}}\abs{\Hset_{A_2}}}  \text{op}_{A_1 A_2 \to A J }(\phi_{A_1 A_2 AJ}) \,,
\end{align}
and define the quantum channels $\Tset^{U^{(2)}}_{A_1^n\to E^n J^n \bar{K} G_B}$, $\Tset^{U^{(1)}}_{A_2^n \to E^n J^n K }$, and $\Tset_{A_1 A_2\to EJ}$, by
\begin{align}
\Tset^{U^{(2)}}_{A_1^n\to E^n J^n \bar{K} G_B}(\rho_{A_1^n})
&
= 
\trace_{B^n}\left[ \Uset_{A\to BE}^{\channel}( \Pi^{U^{(2)}}_{A_1^n\to A^n J^n \bar{K} G_B}(\rho_{A_1^n})) \right] \,,
\label{eq:TU2channel}
\\
\Tset^{U^{(1)}}_{A_2^n \to E^n J^n K }(\rho_{A_2^n })
&
= 
\trace_{B^n}\left[ \Uset_{A\to BE}^{\channel}(\Pi^{U^{(1)}}_{A_2^n \to A^n J^n K}(\rho_{A_2^n })) \right] \,,
\label{eq:TU1channel}
\\
\Tset_{A_1 A_2\to EJ }(\rho_{A_1 A_2})
&
= 
\trace_B\left[ \Uset_{A\to BE}^{\channel}(\Pi_{A_1 A_2\to AJ}(\rho_{A_1 A_2})) \right]
 \,.
\end{align}
According to \cite[Lemma 2.7]{Dupuis:10z}, $\text{op}_{A\to B}(\psi_{AB})\ket{\phi_{AC}}=\text{op}_{A\to C}(\phi_{AC})\ket{\psi_{AB}}$, hence
\begin{align}
&\Tset^{U^{(2)}}_{A_1^n\to E^n J^n \bar{K} G_B}(U^{(1)}_{A_1^n}\psi^{(1)}_{A_1^n K})
=\Tset^{U^{(1)}}_{A_2^n \to E^n J^n K }(U^{(2)}_{A_2^n }\psi^{(2)}_{A_2^n  \bar{K} G_B})\nonumber\\
&=\Tset_{A_1 A_2\to EJ}^{\otimes n}(U^{(1)}_{A_1^n}\psi^{(1)}_{A_1^n K}\otimes U^{(2)}_{A_2^n }\psi^{(2)}_{A_2^n  \bar{K} G_B}) \ \,.
\label{eq:channelEQ}
\end{align}

Applying the decoupling theorem, Theorem~\ref{theo:decoup1}, to the channels  in (\ref{eq:TU2channel})-(\ref{eq:TU1channel}), we obtain
\begin{align}
&\int_{\mathbb{U}_{A_1^n}} \big\lVert \Tset^{U^{(2)}}_{A_1^n\to E^n J^n \bar{K} G_B}(U^{(1)}_{A_1^n}\psi^{(1)}_{A_1^n K})
-\theta_{K}\otimes 
\omega^{U^{(2)}}_{E^n J^n \bar{K} G_B}
\big\rVert_1 \, d U^{(1)}_{A_1^n}  
\leq
2^{-\frac{1}{2}\left[ H_{\min}^{\eps}(A_1^n|  E^n J^n \bar{K} G_B)_{\omega^{U^{(2)}}}-nQ \right]} +12\eps \,,
\label{eq:Dbound1}
\\
&
\int_{\mathbb{U}_{A_2^n }} \big\lVert \Tset^{U^{(1)}}_{A_2^n \to E^n J^n K }(U^{(2)}_{A_2^n }\psi^{(2)}_{A_2^n  \bar{K} G_B})
-\xi_{\bar{K}}\otimes 
\omega^{U^{(1)}}_{E^n J^n K}
\big\rVert_1 \, d U^{(2)}_{A_2^n }  
\leq
2^{-\frac{1}{2}\left[ H_{\min}^{\eps}(A_2^n |  E^n J^n K)_{\omega^{U^{(1)}}}-n(Q+Q'-R_e) \right]} +12\eps \,.
\label{eq:Dbound2}
\end{align}
Using (\ref{eq:channelEQ}), we can rewrite those decoupling inequalities as
\begin{align}
&
\int_{\mathbb{U}_{A_1^n}} \big\lVert \Tset_{A_1 A_2 \to E D }^{\otimes n} (U^{(1)}_{A_1^n}\psi^{(1)}_{A_1^n K}\otimes U^{(2)}_{A_2^n }\psi^{(2)}_{A_2^n  \bar{K} G_B})
-\theta_{K}\otimes 
\omega^{U^{(2)}}_{E^n J^n \bar{K} G_B}
\big\rVert_1 \, d U^{(1)}_{A_1^n} 
\leq
2^{-\frac{1}{2}\left[ H_{\min}^{\eps}(A_1^n|  E^n J^n \bar{K} G_B)_{\omega^{U^{(2)}}}-n(Q-\alpha_{1n}) \right]}  \,,
\label{eq:Dbound3}
\\
&\int_{\mathbb{U}_{A_2^n }} \big\lVert \Tset_{A_1 A_2\to EJ}^{\otimes n}(U^{(1)}_{A_1^n}\psi^{(1)}_{A_1^n K}\otimes U^{(2)}_{A_2^n }\psi^{(2)}_{A_2^n  \bar{K} G_B})
-\xi_{\bar{K}}\otimes 
\omega^{U^{(1)}}_{E^n J^n K}
\big\rVert_1 \, d U^{(2)}_{A_2^n }
  \leq 2^{-\frac{1}{2}\left[ H_{\min}^{\eps}(A_2^n |  E^n J^n K)_{\omega^{U^{(1)}}}-n(Q+Q'-R_e+\alpha_{2n}) \right]}  \,,
\label{eq:Dbound4}
\end{align}
where $\alpha_{jn}\to 0$ as $n\to\infty$. The last bounds tend to zero provided that
\begin{align}
Q&< \frac{1}{n} H_{\min}^{\eps}(A_1^n|  E^n J^n \bar{K} G_B)_{\omega^{U^{(2)}}}-\alpha_{1n} \,, \label{eq:achQ1} \\
Q+Q'-R_e &< \frac{1}{n} H_{\min}^{\eps}(A_2^n |  E^n J^n K)_{\omega^{U^{(1)}}}-\alpha_{2n} \,. \label{eq:achQ2}
\end{align}

This is close to what we would like to show. However, we need the encoder to be an isometry, and we need to replace 
 $\omega^{U^{(1)}}$, $\omega^{U^{(2)}}$ in the inequalities above by $\omega$.
Applying the decoupling theorem, with $\bar{\Tset}^{U^{(2)}}_{A_1^n\to \bar{K} G_B}(\rho_{A_1^n})=\trace_{A^n J^n }[\Pi^{U^{(2)}}_{A_1^n\to A^n J^n \bar{K} G_B}(\rho_{A_1^n})]$ and $\bar{\Tset}_{A_2^n \to \mathbb{C}}(\rho_{A_2^n })=\trace[\Pi_{A_2^n \to A_1^n A^n J^n}(\rho_{A_2^n })]$, we obtain
\begin{align}
&
\int_{\mathbb{U}_{A_1^n}} \Big\lVert \trace_{A^n J^n }\big[\Pi^{U^{(2)}}_{A_1^n\to A^n J^n \bar{K} G_B}%
U^{(1)}_{A_1^n}\psi^{(1)}_{A_1^n K}\big]
-\theta_{K}\otimes 
\omega^{U^{(2)}}_{\bar{K} G_B}
\Big\rVert_1 \, d U^{(1)}_{A_1^n}
\leq
2^{-\frac{1}{2}\left[ H_{\min}^{\eps}(A_1^n|  \bar{K} G_B)_{\omega^{U^{(2)}}}-n(Q+\alpha_{3n}) \right]}  \,,
\label{eq:Dbound5}
\\
&\int_{\mathbb{U}_{A_2^n }} \big\lVert \omega^{U^{(2)}}_{\bar{K} G_B}-\xi_{\bar{K}}\otimes \Phi_{G_B}
\big\rVert_1 \, d U^{(2)}_{A_2^n }
=\int_{\mathbb{U}_{A_2^n }} \big\lVert \bar{\Tset}_{A_2^n \to \mathbb{C}}(U^{(2)}_{A_2^n }\psi^{(2)}_{A_2^n  \bar{K} G_B})
-\xi_{\bar{K}}\otimes \Phi_{ G_B}
\big\rVert_1 \, d U^{(2)}_{A_2^n }  
\nonumber\\
&\leq
2^{-\frac{1}{2}n[ H(A_2)_{\phi}-(Q+Q'+R_e+\alpha_{4n}) ]}  \,,
\label{eq:Dbound6}
\end{align}
which tend to zero for 
\begin{align}
Q&< \frac{1}{n} H_{\min}^{\eps}(A_1^n|  \bar{K} G_B)_{\omega^{U^{(2)}}}-\alpha_{3n} \,, \label{eq:achQ3} \\
Q+Q'+R_e&< H(A_2)_{\phi}-\alpha_{4n} \label{eq:achQ4} \,.
\end{align}
We will use these decoupling properties in order to obtain the encoding map, by applying  Uhlmann's theorem (see Theorem~\ref{theo:Uhlmann}). However, before that, we give the bounds on the smoothed min-entropies.

\subsubsection{Entropy Bounds}
We would like to bound the min-entropies in (\ref{eq:achQ1})-(\ref{eq:achQ2}) and (\ref{eq:achQ3}). %
We begin with the min-entropy $H_{\min}^{\eps}(A_1^n|  \bar{K} G_B)_{\omega^{U^{(2)}}}$.
Suppose that $\tilde{\phi}_{A_1^n A_2^n  A^n J^n}$ is a state such that
\begin{align}
&\left\lVert\tilde{\phi}_{A_1^n A_2^n  A^n J^n}-\phi^{\otimes n}_{A_1 A_2  A J} \right\rVert_1 \leq 2\eps_0
\label{eq:sigmaDist}
\intertext{and}
& H_{\min}(A_1^n|  A_2^n )_{\tilde{\phi}}= H_{\min}(A_1^n|  A_2^n )_{\phi} \,.
\end{align}
Define %
\begin{align}
&\tilde{\Pi}^{U^{(2)}}_{A_2^n \to \bar{K} G_B}= 
\sqrt{\abs{\Hset_{A_2^n }}}
\text{op}_{A_2^n \to \bar{K} G_B}( U^{(2)}_{A_2^n} W^{(2)}_{\bar{M}G_A\to A_2^n} |\psi^{(2)}_{\bar{M} G_A G_B \bar{K}} \rangle)%
\,,
\end{align}
hence
\begin{align}
 \ket{\omega_{A_1^n A^n J^n \bar{K} G_B}^{U^{(2)}}}=\tilde{\Pi}^{U^{(2)}}_{A_2^n \to \bar{K} G_B}
\ket{\phi^{\otimes n}_{A_1 A_2 A J  }} \,.
\end{align}

Consider a decomposition of the operator $(\tilde{\phi}-\phi^{\otimes n})$ into its positive and negative parts,
\begin{align}
\tilde{\phi}_{A_1^n A_2^n  A^n J^n}-\phi^{\otimes n}_{A_1 A_2  A J}=\Delta_+ - \Delta_- \,,
\end{align}
where $\Delta_+$, $ \Delta_-\succeq 0$ have a disjoint support. Hence, by (\ref{eq:sigmaDist}), $\trace(\Delta_\pm)\leq 2\eps_0$.
We note that 
\begin{align}
&\int_{\mathbb{U}_{A_2^n }} \tilde{\Pi}^{U^{(2)}}_{A_2^n \to \bar{K} G_B}(P_{A_2^n } ) \, dU^{(2)}_{A_2^n }
=\trace(P_{A_2^n })\cdot 
\frac{1}{\abs{\Hset_{A_2}}^n}\tilde{\Pi}^{U^{(2)}}_{A_2^n \to \bar{K} G_B}\tilde{\Pi}^{U^{(2)}\dagger}_{A_2^n \to \bar{K} G_B}
\,.
\end{align}
Thus,
\begin{align}
\int_{\mathbb{U}_{A_2^n }} \Big\lVert \tilde{\Pi}_{A_2^n \to \bar{K} G_B}^{U^{(2)}}(\tilde{\phi}_{A_1^n A_2^n  A^n J^n})
-\tilde{\Pi}^{U^{(2)}}_{A_2^n \to \bar{K} G_B}(\phi^{\otimes n}_{A_1 A_2  A J})
 \Big\rVert_1 \, dU^{(2)}_{A_2^n}
&=
\int_{\mathbb{U}_{A_2^n }} \left\lVert \tilde{\Pi}^{U^{(2)}}_{A_2^n \to \bar{K} G_B}(\Delta_+ - \Delta_-)
 \right\rVert_1 \, dU^{(2)}_{A_2^n}
\nonumber\\
&\leq  \trace\left(\tilde{\Pi}^{U^{(2)}}_{A_2^n \to \bar{K} G_B}(\Delta_+)\right)+
\trace\left(\tilde{\Pi}^{U^{(2)}}_{A_2^n \to \bar{K} G_B}(\Delta_-)\right)
\nonumber\\
&\leq 4\eps_0 \,,
\end{align}
which, in turn, implies
\begin{align}
&\int_{\mathbb{U}_{A_2^n }} d_F\Big( \tilde{\Pi}^{U^{(2)}}_{A_2^n \to \bar{K} G_B}(\tilde{\phi}_{A_1^n A_2^n  A^n J^n})
\,,\;  
\tilde{\Pi}^{U^{(2)}}_{A_2^n \to \bar{K} G_B}(\phi^{\otimes n}_{A_1 A_2  A J})
 \Big) \, dU^{(2)}_{A_2^n}
\leq 2\sqrt{\eps_0} \,,
\end{align}
based on the relation between the fidelity and trace distance \cite[Corollary 9.3.1]{Wilde:17b}.
We deduce that 
\begin{align}%
\int_{\mathbb{U}_{A_2^n }} H_{\min}^{2\sqrt{\eps_0}}(A_1^n|  \bar{K} G_B)_{\omega^{U^{(2)}}} \, dU^{(2)}_{A_2^n}\geq
 H_{\min}^{\eps_0}(A_1^n|  A_2^n )_{\phi} \,.
\end{align}%
For $\eps_0=\frac{\eps^2}{4}$, we obtain
\begin{align}%
\int_{\mathbb{U}_{A_2^n }} \frac{1}{n} H_{\min}^{\eps}(A_1^n|  \bar{K} G_B)_{\omega^{U^{(2)}}} \, dU^{(2)}_{A_2^n} \geq 
 H(A_1|  A_2 )_{\phi}-\alpha_{5n} \,.
\label{eq:HA1gA2}
\end{align}%

By the same considerations,
\begin{align}
\int_{\mathbb{U}_{A_1^n }} \frac{1}{n} H_{\min}^{\eps}(A_2^n |  E^n J^n K)_{\omega^{U^{(1)}}} \, dU^{(1)}_{A_1^n}
 &\geq 
H(A_2 |  E D A_1)_{\omega}-\alpha_{6n} 
\nonumber\\
&= -H(A_2|  B)_\omega-\alpha_{6n}
\nonumber\\
&= I(A_2\rangle B)_\omega-\alpha_{6n}
\,,
\label{eq:HA1gEM}
\end{align}
and
\begin{align}
\int_{\mathbb{U}_{A_2^n }} \frac{1}{n} H_{\min}^{\eps}(A_1^n|  E^n J^n \bar{K} G_B)_{\omega^{U^{(2)}}} \, dU^{(2)}_{A_2^n}
 &\geq 
H(A_1|  E D A_2 )_{\omega}-\alpha_{7n}
\nonumber\\
&= I(A_1\rangle B)_\omega-\alpha_{7n}
 \,.
\label{eq:HA2rB}
\end{align}

Therefore, there exist unitaries $U^{(1)}_{A_1^n}$ and $U^{(2)}_{A_2^n }$ that satisfy the following inequalities,
\begin{align}
\Big\lVert \trace_{A^n J^n }\big[\Pi^{U^{(2)}}_{A_1^n\to A^n J^n \bar{K} G_B}%
U^{(1)}_{A_1^n}\psi^{(1)}_{A_1^n K}\big]
-\theta_{K}\otimes 
\omega^{U^{(2)}}_{\bar{K} G_B}
\Big\rVert_1 
 &\leq
2^{-\frac{1}{2}n [ H(A_1|  A_2)_{\phi}-Q-\beta_n ]} \,,
\label{eq:Dbound10}
\\
\big\lVert \Tset_{A_1 A_2\to EJ}^{\otimes n}(U^{(1)}_{A_1^n}\psi^{(1)}_{A_1^n K}\otimes U^{(2)}_{A_2^n }\psi^{(2)}_{A_2^n  \bar{K} G_B})
-\theta_{K}\otimes 
\omega^{U^{(2)}}_{E^n J^n \bar{K} G_B}
\big\rVert_1  
 &\leq
 2^{-\frac{1}{2}n[ I(A_1\rangle B)_{\omega}-Q-\beta_n ]}  \,,
\label{eq:Dbound8}
\\
%
\big\lVert \Tset_{A_1 A_2\to EJ}^{\otimes n}(U^{(1)}_{A_1^n}\psi^{(1)}_{A_1^n K}\otimes U^{(2)}_{A_2^n }\psi^{(2)}_{A_2^n  \bar{K}})
-\xi_{\bar{K}}\otimes 
\omega^{U^{(1)}}_{E^n J^n K}
\big\rVert_1  
&\leq
 2^{-\frac{1}{2}n[ I(A_2 \rangle B)_{\omega}-(Q+Q'-R_e)-\beta_n ]}  \,,
\label{eq:Dbound9}
\\
%
 \big\lVert \omega^{U^{(2)}}_{\bar{K} G_B}-\xi_{\bar{K}}\otimes \Phi_{G_B}
\big\rVert_1 
=\big\lVert \bar{\Tset}_{A_2\to \mathbb{C}}^{\otimes n}(U^{(2)}_{A_2^n }\psi^{(2)}_{A_2^n  \bar{K} G_B})
-\xi_{\bar{K}}\otimes \Phi_{G_B}
\big\rVert_1  
&\leq
 2^{-\frac{1}{2}n[ H(A_2 )_{\phi}-(Q+Q'+R_e)-\beta_n ]} %
\,,
\label{eq:Dbound11}
\end{align}
for some $\beta_n$ that tends to zero as $n\to \infty$, as the first inequality follows from (\ref{eq:Dbound5}) and (\ref{eq:HA1gA2}), the second 
from  (\ref{eq:Dbound3}) and (\ref{eq:HA2rB}), the third is due to (\ref{eq:Dbound4}) and (\ref{eq:HA1gEM}), and the last holds by (\ref{eq:Dbound6}).

\subsubsection{Encoding}
By the triangle inequality, (\ref{eq:Dbound10}) and (\ref{eq:Dbound11}) %
yield
\begin{align}
& \Big\lVert \trace_{A^n J^n }\big[\Pi^{U^{(2)}}_{A_1^n\to A^n J^n \bar{K} G_B}%
U^{(1)}_{A_1^n}\psi^{(1)}_{A_1^n K}\big]
-\theta_{K}\otimes 
\xi_{\bar{K}}\otimes 
\Phi_{G_B}
\Big\rVert_1
\leq
\delta_{\text{enc}}(n)
 \,,
\label{eq:Dbound13}
\intertext{where}
&\delta_{\text{enc}}(n) \equiv 
2^{-\frac{1}{2}n [ H(A_1|  A_2)_{\phi}-Q-\beta_n ]} 
+
2^{-\frac{1}{2}n[ H(A_2 )_{\phi}-(Q+Q'+R_e)-\beta_n ]} \,.
\end{align}
Based on Uhlmann's theorem, it follows that there exists an isometry $F_{M \bar{M} G_A\to A^n J^n }$ such that
\begin{align}
& \Big\lVert \Pi^{U^{(2)}}_{A_1^n\to A^n J^n \bar{K} G_B}%
U^{(1)}_{A_1^n}\psi^{(1)}_{A_1^n K}
-F_{M \bar{M} G_A\to A^n J^n }  (\theta_{M K}\otimes 
\xi_{\bar{M} \bar{K} }\otimes 
\Phi_{ G_A  G_B})
\Big\rVert_1
\leq
2\sqrt{\delta_{\text{enc}}(n)}
 \,.
\label{eq:Dbound14}
\end{align}

\subsubsection{Decoding without Assistance}
By applying the isometric extension of the channel to the states on the LHS of (\ref{eq:Dbound14}) and using the triangle inequality and the monotonicity
of the trace distance under quantum channels, we obtain
\begin{align}
& \Big\lVert \Tset^{U^{(2)}}_{A_1^n\to E^n J^n \bar{K} G_B}(
U^{(1)}_{A_1^n}\psi^{(1)}_{A_1^n K})
-\trace_{B^n}\big[ (U_{A\to BE}^{\channel})^{\otimes n} F_{M \bar{M} G_A\to A^n J^n }
(\theta_{M K}\otimes 
\xi_{\bar{M}  \bar{K} }\otimes 
\Phi_{G_A G_B})\big]
\Big\rVert_1
\leq
2\sqrt{\delta_{\text{enc}}(n)}
 \,.
\label{eq:Dbound15}
\end{align}
By (\ref{eq:channelEQ}), the inequality above can also be written as
\begin{align}
& \Big\lVert \Tset_{A_1 A_2\to EJ}^{\otimes n}(U^{(1)}_{A_1^n}\psi^{(1)}_{A_1^n K}\otimes U^{(2)}_{A_2^n }\psi^{(2)}_{A_2^n  \bar{K} G_B})
-\trace_{B^n}\big[ (U_{A\to BE}^{\channel})^{\otimes n} F_{M \bar{M} G_A\to A^n J^n }
(\theta_{M K}\otimes 
\xi_{\bar{M} \bar{K} }\otimes 
\Phi_{G_A G_B})\big]
\Big\rVert_1
\leq
2\sqrt{\delta_{\text{enc}}(n)}
 \,.
\label{eq:Dbound16}
\end{align}
Together with (\ref{eq:Dbound8}), %
this implies
\begin{align}
& \Big\lVert \trace_{B^n}\big[ (U_{A\to BE}^{\channel})^{\otimes n} F_{M \bar{M} G_A\to A^n J^n }
(\theta_{M K}\otimes 
\xi_{\bar{M} \bar{K} }\otimes 
\Phi_{ G_A  G_B})\big]
-\theta_{K}\otimes 
\omega^{U^{(2)}}_{E^n J^n \bar{K} G_B}
\Big\rVert_1
\leq
2\sqrt{\delta_{\text{enc}}(n)}+\delta_1(n)
 \,,
\label{eq:Dbound17}
\end{align}
where $\delta_1(n)\equiv 2^{-\frac{1}{2}n[ I(A_1\rangle B)_{\omega}-Q-\beta_n ]}$.

Thus, by Uhlmann's theorem, there exists an isometry $D^*_{B^n \to M J_1}$, such that 
\begin{align}
& \Big\lVert D^*_{B^n\to M J_1} (U_{A\to BE}^{\channel})^{\otimes n} F_{M \bar{M} G_A\to A^n J^n }
(\theta_{M K}\otimes 
\xi_{\bar{M} \bar{K} }\otimes 
\Phi_{ G_A  G_B})
-\theta_{M K}\otimes 
\hat{\omega}_{E^n J^n \bar{K} G_B J_1}
\Big\rVert_1
\leq
2\sqrt{2\sqrt{\delta_{\text{enc}}(n)}+\delta_1(n)}
 \,,
\label{eq:Dbound18}
\end{align}
where $\hat{\omega}_{E^n J^n \bar{K} G_B J_1}$ is an arbitrary purification of $\omega^{U^{(2)}}_{E^n J^n \bar{K} G_B}$.
By tracing over $E^n J^n \bar{K}G_B J_1$, we deduce that there exist an encoding map $\Fset_{M\bar{M} G_A\to A^n }$ and a decoding map $\Dset^*_{B^n\to M}$, such that 
\begin{align}
& \Big\lVert (\Dset^*_{B^n\to M J_1}\circ \channel^{\otimes n}_{A\to B}\circ \Fset_{M \bar{M} G_A\to A^n })
(\theta_{M K}\otimes 
\xi_{\bar{M} }\otimes 
\Phi_{ G_A})-\theta_{M K}
\Big\rVert_1
\leq
2\sqrt{2\sqrt{\delta_{\text{enc}}}+\delta_1(n)}
 \,.
\label{eq:Dbound18t}
\end{align}

\subsubsection{Decoding with Entanglement Assistance}
The bound in (\ref{eq:Dbound16}),  together with (\ref{eq:Dbound9}), %
implies
\begin{align}
& \Big\lVert \trace_{B^n G_B}\big[ (U_{A\to BE}^{\channel})^{\otimes n} F_{M \bar{M} G_A\to A^n J^n }
(\theta_{M K}\otimes 
\xi_{\bar{M}  \bar{K} }\otimes 
\Phi_{G_A G_B})\big]
-\xi_{\bar{K}}\otimes 
\omega^{U^{(1)}}_{E^n J^n K}
\Big\rVert_1
\leq
2\sqrt{\delta_{\text{enc}}(n)}+\delta_2(n)
 \,,
\label{eq:Dbound19}
\end{align}
where $\delta_2(n)\equiv 2^{-\frac{1}{2}n[ I(A_2 \rangle B)_{\omega}-(Q+Q'-R_e)-\beta_n ]} $.
Thus, by Uhlmann's theorem, there exists an isometry $D_{B^n G_B\to \bar{M} G'_A G'_B J_2}$, such that 
\begin{align}
& \Big\lVert D_{B^n G_B\to  \bar{M} G_A' G_B' J_2} (U_{A\to BE}^{\channel})^{\otimes n} F_{M \bar{M} G_A\to A^n J^n }
(\theta_{M K}\otimes 
\xi_{\bar{M} \bar{K}}\otimes 
\Phi_{ G_A  G_B})
-\xi_{\bar{M} \bar{K}}\otimes 
\Phi_{ G_A  G_B}\otimes 
\hat{\omega}%
_{E^n J^n K J_2}
\Big\rVert_1
\nonumber\\&
\leq
2\sqrt{2\sqrt{\delta_{\text{enc}}(n)}+\delta_2(n)}
 \,.
\label{eq:Dbound20}
\end{align}
By tracing over $E^n J^n K G_A' G_B'  J_2$, we deduce that $\Fset_{M\bar{M} G_A\to A^n }$ and $\Dset_{BG_B \to \bar{M}}$ satisfy 
\begin{align}
& \Big\lVert \Dset_{B^n G_B\to  \bar{M} }\circ \channel_{A\to B}^{\otimes n}\circ 
\Fset_{M \bar{M} G_A\to A^n }
(\theta_{M}\otimes 
\xi_{\bar{M} \bar{K}}\otimes 
\Phi_{ G_A  G_B})
-\xi_{\bar{M} \bar{K}}
\Big\rVert_1
\leq
2\sqrt{2\sqrt{\delta_{\text{enc}}(n)}+\delta_2(n)}
 \,.
\label{eq:Dbound20t}
\end{align}
As $\delta_{\text{enc}}(n)$, $\delta_1(n)$, and $\delta_2(n)$ tend to zero as $n\to\infty$ for rates as in (\ref{eq:QrateAch}), we deduce that the errors tend to zero as well. 

Choosing the entanglement rate $R_e=\frac{1}{2}[H(A_2)_\omega+H(A_2|  B)_\omega]$, it follows that  $\inQ_{\text{EA}*}(\channel)$ is an achievable rate region (\cf (\ref{eq:calRQea}) and (\ref{eq:QrateAch})).
To show that rate pairs in $\frac{1}{k}  \inQ_{\text{EA}*}(\channel^{\otimes k})$ are achievable as well, we employ the coding scheme above for the product channel $\channel^{\otimes k}$, where $k$ is arbitrarily large.
This completes the achievability proof. 

\subsection{Converse Proof}
\label{app:MskQ}
Consider the converse part. 
Suppose that Alice and Bob are trying to generate entanglement between them. %
An upper bound on the rate at which Alice and Bob can generate entanglement also serves as an upper bound on the quantum rate at which they can communicate qubits, since a noiseless quantum channel can be used to generate entanglement by sending one part of an entangled pair. In this task, Alice locally prepares two maximally entangled pairs,
\begin{align}
&|\Phi_{MK }\rangle \otimes |\Phi_{\bar{M} \bar{K}}\rangle = \left(\frac{1}{\sqrt{2^{nQ}}}\sum_{m=1}^{2^{nQ}} | m \rangle_{M} \otimes | m \rangle_{K} \right)
\otimes
\left(\frac{1}{\sqrt{2^{n(Q+Q')}}}\sum_{\bar{m}=1}^{2^{n(Q+Q')}} | \bar{m} \rangle_{\bar{M}} \otimes | \bar{m} \rangle_{\bar{K}} \right)
 \,.
\end{align}
If the entanglement assistance is reliable, then Alice and Bob share the quantum state $|\Phi_{G_A G_B }\rangle$, where $G_A$ and $G_B$ represent the entanglement resources that Alice and Bob share, respectively.
 Then Alice applies an encoding channel $\Fset_{M \bar{M} G_A \to A^n}$ to the quantum message systems $M\bar{M}$ and  her share $G_A$ of the entanglement resources. 
The resulting state is 
\begin{align}%
\varphi_{K\bar{K}     G_B A^n}\equiv
 \Fset_{M \bar{M} G_A \to A^n}( \Phi_{MK } \otimes \Phi_{\bar{M} \bar{K}}\otimes \Phi_{G_A G_B }  
) \,.
\label{eq:QconvI1}
\end{align}%
As Alice sends the systems $A^n$ through the channel, the output state is
\begin{align}
\omega_{K\bar{K}     G_B B^n}\equiv \channel^{\otimes n}_{A\to B} (\sigma_{ K\bar{K}     G_B A^n}) \,.
\end{align}
If the entanglement assistance is present, then Bob can access $G_B$. In this case, Bob performs a decoding channel $\Dset_{G_B B^n \rightarrow \breve{M}}$, hence
\begin{align}
\rho_{ K \bar{K} \breve{M} }\equiv \Dset_{G_B B^n \rightarrow \breve{M}}(\rho_{K\bar{K}     G_B B^n}) \,.
\label{eq:DecConv1QEA}
\end{align}
On the other hand, without assistance,   Bob performs  $\Dset^*_{B^n \rightarrow \hM}$, producing 
\begin{align}
\rho^*_{ K \bar{K} \hM  G_B}\equiv \Dset^*_{B^n \rightarrow \hM}(\rho_{K\bar{K}     G_B B^n}) \,.
\label{eq:DecConv1Q}
\end{align}
Since Bob has not received the entanglement resources, the system $G_B$  is not affected by itself.

Consider a sequence of codes $(\Fset_n,\Dset_n,\Dset_n^*)$ for entanglement generation given unreliable assistance, such that
\begin{align}
\frac{1}{2} \norm{ \rho_{\breve{M} \bar{K}} -\Phi_{\bar{M} \bar{K}} }_1 \leq& \alpha_n \,, \label{eq:randDconvQEA} \\
\frac{1}{2} \norm{ \rho^*_{\hM K} -\Phi_{M K} }_1 \leq& \alpha^*_n \,, \label{eq:randDconvQ} 
\end{align}
where %
$\alpha_n,\alpha^*_n$ tend to zero as $n\rightarrow\infty$.

By the Alicki-Fannes-Winter inequality \cite{AlickiFannes:04p,Winter:16p} \cite[Theorem 11.10.3]{Wilde:17b}, (\ref{eq:randDconvQ}) implies $|H(K|\hM)_{\rho^*} - H(K|M)_{\Phi} |\leq n\eps_n$, or equivalently,
\begin{align}
|I(K \rangle \hM)_{\rho^*} - I(K \rangle M)_{\Phi} |\leq n\eps_n \,,
\label{eq:AFWq}
\end{align}
where $\eps_n$ tends to zero as $n\rightarrow\infty$. 
Observe that $I(K \rangle M)_\Phi=H(M)_\Phi-H(KM)_\Phi=nQ-0=nQ$. Thus,
\begin{align}
nQ=& I(K\rangle M)_\Phi \nonumber\\
  \leq& I(K\rangle \hM)_{\rho^*}+n\eps_n \nonumber\\
	\leq& I(K\rangle B^n)_\omega+n\eps_n \,,
	\label{eq:ineqQuna1}
\end{align}
where the last line follows from (\ref{eq:DecConv1Q}) and the data processing inequality for the coherent information 
\cite[Theorem 11.9.3]{Wilde:17b}.
In addition,
\begin{align}
nQ=& H(K)_\Phi \nonumber\\
  =& H(K|\bar{K}G_B)_{\Phi\otimes\Phi\otimes\Phi} \nonumber\\
	=& H(K|\bar{K}G_B)_\omega \,,
	\label{eq:ineqQuna2}
\end{align}
where the second line follows since $K$, $\bar{K}$, $G_B$ are in a product state.

As for decoding with entanglement assistance,
(\ref{eq:randDconvQEA}) implies $|I(\bar{K};\bar{M})_{\Phi}-I(\bar{K};\breve{M})_{\rho}   |\leq n \bar{\eps}_n$, %
by the Alicki-Fannes-Winter inequality, where $\bar{\eps}_n$ tends to zero as $n\rightarrow\infty$. 
Therefore,
\begin{align}
n(Q+Q')&=\frac{1}{2} I(\bar{K};\bar{M})_\Phi \nonumber\\
&\leq \frac{1}{2} I(\bar{K};\breve{M})_\rho+n\bar{\eps}_n \nonumber\\
&\leq \frac{1}{2} I(\bar{K};G_B B^n)_\omega+n\bar{\eps}_n
\label{eq:ineqQuna3a}
\end{align}
by %
the data processing inequality for the quantum mutual information.
By the chain rule, the mutual information above satisfies
\begin{align}
I(\bar{K};G_B B^n)_\omega&= I(\bar{K}; B^n|G_B)_\omega+I(\bar{K};G_B)_{\Phi\otimes\Phi} \nonumber\\
&= I(\bar{K}; B^n|G_B)_\omega \nonumber\\
&\leq I(\bar{K}G_B; B^n)_\omega \,.
\label{eq:ineqQuna3b}
\end{align}

The regularized converse follows from
(\ref{eq:ineqQuna1}), (\ref{eq:ineqQuna2}), and (\ref{eq:ineqQuna3a})-(\ref{eq:ineqQuna3b}),
as we let $A_1^n$ and $A_2^n$ be quantum systems such that for some isometries $W^{(1)}_{K\rightarrow A_1^n}$ and $W^{(2)}_{\bar{K}G_B\rightarrow A_2^n}$, 
\begin{align}
&\varphi_{A_1^n A_2^n  A^n}=
\left(W^{(1)}_{K\rightarrow A_1^n}\otimes W^{(2)}_{\bar{K}G_B\rightarrow A_2^n} \right)
 \varphi_{K\bar{K}G_B  A^n   } \left(W^{(1)}_{K\rightarrow A_1^n}\otimes W^{(2)}_{\bar{K}\rightarrow A_2^n} \right)^\dagger \,.
\end{align}
This completes the proof of Theorem~\ref{theo:qEA}. %
\qed

}
\end{appendices}

\bibliography{references}{}
\begin{IEEEbiographynophoto}
{Uzi Pereg} 
 received the B.Sc. (summa cum laude) degree in electrical engineering from Azrieli College of Engineering, Jerusalem, Israel, in 2011, and the M.Sc. and Ph.D. degrees from the Technion -- Israel Institute of Technology, Haifa, Israel, in 2015 and 2019, respectively. He was a postdoctoral fellow at the Institute for Communications Engineering of the Technical University of Munich and the Munich Center for Quantum Science and Technology (MCQST), from 2020 to 2022. During this period, he participated in the theory group of the German federal government (BMBF) project for the design and analysis of quantum communication and repeater systems. Currently, he is an associate professor in the Faculty of Electrical and Computer Engineering and the Hellen Diller Quantum Center, Technion -- Israel Institute of Technology, Haifa, Israel (uzipereg@technion.ac.il). His research interests are in the areas of quantum communications, information theory, and coding theory. Uzi is a recipient of the 2018 Pearl award for outstanding research work in the field of communications, 2018 KLA-Tencor Award for Excellent Conference Paper, 2019 Viterbi Ph.D. Fellowship, 2020 Viterbi postdoctoral fellowship, 2020-2021 Israel CHE Fellowship for Quantum Science and Technology, 2022 Seed Funding Grant for Exceptional Junior Researchers of the Munich Center for Quantum Science and Technology, 2022 Israel VATAT Junior Faculty Program for Quantum Science and Technology, and the Chaya Career Advancement Chair.
\end{IEEEbiographynophoto}

\begin{IEEEbiographynophoto}
{Christian Deppe} 
 received the Dipl.-Math. degree in mathematics from the Universit\"at Bielefeld, Bielefeld, Germany, in 1996, and the Dr.-Math. degree in mathematics from the Universit\"at Bielefeld, Bielefeld, Germany, in 1998. He was a Research and Teaching Assistant with the Fakult\"at f\"ur Mathematik, Universit\"at Bielefeld, from 1998 to 2010. From 2011 to 2013 he was the leader of the project ``Sicherheit und Robustheit des Quanten-Repeaters´´ of the Federal Ministry of Education and Research at Fakult\"at f\"ur Mathematik, Universit\"at Bielefeld. In 2014 he was supported by a DFG project at the Institute of Theoretical Information Technology, Technische Universit\"at M\"unchen. In 2015 he held a temporary professorship at the Fakult\"at f\"ur Mathematik und Informatik, Friedrich-Schiller Universit\"at Jena. Since 2018 he is at the Department of Communications Engineering at the Technical University of Munich. He is the leader of several projects funded by the BMBF, the DFG, and the state of Bavaria.
\end{IEEEbiographynophoto}

\begin{IEEEbiographynophoto}
{Holger Boche}  (Fellow, IEEE) received the Dipl.-Ing. degree in electrical
engineering, the Graduate degree in mathematics, and the Dr.-Ing. degree
in electrical engineering from the Technische Universit\"at Dresden, Dresden,
Germany, in 1990, 1992, and 1994, respectively, the master's degree from
the Friedrich-Schiller Universit\"at Jena, Jena, Germany, in 1997, and the
Dr. rer. nat. degree in pure mathematics from the Technische Universit\"at
Berlin, Berlin, Germany, in 1998.
In 1997, he joined the Fraunhofer Institute for Telecommunications, Heinrich Hertz Institute (HHI), Berlin. From 2002 to 2010, he was a Full Professor
of mobile communication networks with the Institute for Communications
Systems, Technische Universit\"at Berlin. In 2003, he became the Director
of the Fraunhofer German-Sino Laboratory for Mobile Communications,
Berlin. In 2004, he became the Director of the Fraunhofer Institute for
Telecommunications, HHI. He was a Visiting Professor with ETH Z\"urich,
Z\"urich, Switzerland, from 2004 to 2006 (Winter), and with KTH Royal
Institute of Technology, Stockholm, Sweden, in 2005 (Summer). He has been
a member and an Honorary Fellow of the TUM Institute for Advanced Study,
Munich, Germany, since 2014. Since 2018, he has been a Founding Director
of the Center for Quantum Engineering, Technische Universit\"at M\"unchen.
Since 2021, he has been jointly leading the BMBF Research Hub 6G-life with
Frank Fitzek. He is currently a Full Professor with the Institute of Theoretical
Information Technology, Technische Universit\"at M\"unchen, Munich, Germany,
which he joined in October 2010. Among his publications is the recent
book, Information Theoretic Security and Privacy of Information Systems
(Cambridge University Press, 2017).
Prof. Boche is a member of the IEEE Signal Processing Society SPCOM
and SPTM Technical Committees. He was an Elected Member of the German
Academy of Sciences (Leopoldina) in 2008 and to the Berlin Brandenburg
Academy of Sciences and Humanities in 2009. He was a recipient of
the Research Award ``Technische Kommunikation" from the Alcatel SEL
Foundation in October 2003, the ``Innovation Award" from the Vodafone
Foundation in June 2006, and the Gottfried Wilhelm Leibniz Prize from the
Deutsche Forschungsgemeinschaft (German Research Foundation) in 2008.
He was a co-recipient of the 2006 IEEE Signal Processing Society Best Paper
Award and a recipient of the 2007 IEEE Signal Processing Society Best Paper
Award. He was the General Chair of the Symposium on Information Theoretic
Approaches to Security and Privacy at IEEE GlobalSIP 2016.
\end{IEEEbiographynophoto} 
\end{document}